\documentclass{amsart}

\usepackage{a4wide}
\usepackage{amscd}
\usepackage{amssymb}
\usepackage{amsthm}
\usepackage{amsfonts}
\usepackage{amsmath}
\usepackage{latexsym}
\usepackage{graphicx}
\usepackage{setspace}
\usepackage{graphicx, subfigure}

\doublespace

\newtheorem{theorem}{Theorem}
\theoremstyle{plain}

\newtheorem{corollary}[theorem]{Corollary}

\newtheorem{proposition}[theorem]{Proposition}
\newtheorem{remark}[theorem]{Remark}

\numberwithin{equation}{section}
\numberwithin{theorem}{section}

\newcommand{\R}{\ensuremath{\mathbb{R}}}

\newcommand{\diag}{\operatornamewithlimits{diag}}
\newcommand{\Var}{\operatornamewithlimits{Var}}

\allowdisplaybreaks
\title[Swing options in commodity markets]{Swing options in commodity markets: A multidimensional L\'evy diffusion model}
\author[Eriksson]{Marcus Eriksson}
\author[Lempa]{Jukka Lempa}
\author[Nilssen]{Trygve Kastberg Nilssen}
\address[Marcus Eriksson]{\newline Department of Mathematics \newline University of Oslo\newline P.O. Box 1053, Blindern\newline N--0316 Oslo, Norway}
\email[]{mkerikss\@@math.uio.no}
\address[Jukka Lempa]{\newline Centre of Mathematics for Applications \newline University of Oslo\newline P.O. Box 1053, Blindern\newline N--0316 Oslo, Norway}
\email[]{jlempa\@@cma.uio.no}
\address[Trygve Kastberg Nilssen]{\newline Department of Economics and Business Administration\newline University of Agder\newline Serviceboks 422 \newline N-4604 Kristiansand, Norway}
\email[]{trygve.k.nilssen\@@uia.no}

\date{\today}

\begin{document}

\maketitle

\begin{abstract}
We study valuation of swing options on commodity markets when the commodity prices are driven by multiple factors. The factors are modeled as diffusion processes driven by a multidimensional L\'evy process. We set up a valuation model in terms of a dynamic programming problem where the option can be exercised continuously in time. Here, the number of swing rights is given by a total volume constraint. We analyze some general properties of the model and study the solution by analyzing the associated HJB-equation. Furthermore, we discuss the issues caused by the multi-dimensionality of the commodity price model. The results are illustrated numerically with three explicit examples.
\end{abstract}

\section{Introduction}

The purpose of this paper is to propose and analyze a model for valuation of a swing option, see, e.g. \cite{Burger et al}, written on multiple commodities when the commodity spot prices are driven by multiple, potentially non-Gaussian factors. More precisely, the model is formulated as a dynamic programming problem in continuous time. The holder of the option is contracted an amount of a given commodity that can be purchased for a fixed price during the lifetime of the contract. The purchases can be done (that is, the option can be exercised) continuously in time such that contracted rate constraints are fulfilled. This form of contract originates from electricity markets, where they are called flexible load contracts, see, e.g. \cite{Bjorgan et al,Lund Ollmar}. However, this model setting can also fit a traditional swing option with a high number of swing rights and possible exercise times. For example, we can think of a situation in an electricity market where contract is written for a year and holder can exercise on the hour-ahead market. This results into over 8000 possible exercise times, which makes, in particular, Monte-Carlo methods virtually intractable.

During the recent years, there has been a lot of activity on analysis of swing options. Being essentially a multi-strike American or Bermudan option, a natural way to approach swing options is via an optimal multiple stopping problem. In the recent papers \cite{Carmona Touzi,Bender 0}, the theory of optimal multiple stopping is developed in continuous time using sophisticated martingale theory. To compute option prices numerically, they develop appropriate Monte-Carlo methodology. Other methodology for swing option pricing includes forests of trees \cite{Jaillet et al,Hambly et al,Wahab et al,Edoli et al} or stochastic meshes \cite{Marshall}, multi-stage stochastic optimization \cite{Haarbrucker Kuhn}, (quasi-)variational inequalities \cite{Dahlgren,Kiesel et al} and PDE approaches \cite{Kjaer, BLN, Lund Ollmar}. Fundamentally, all of these methods are based on the dynamic programming principle.

As the main contribution of this paper, we develop a valuation model for multi-commodity swing options inspired by \cite{BLN}. In \cite{BLN}, the valuation problem was studied in the case of a single commodity driven by a one-factor Gaussian price process. In this paper, we generalize the results of \cite{BLN} to cover multiple contracted commodities with prices driven by multiple factors. From applications point of view, this is an important generalization, since there is a substantial body of literature supporting the usage of multi-factor models for commodity prices. Moreover, we allow also for non-Gaussian factors, which are favored, for example, in electricity price models, see, e.g. \cite{Benth Kallsen M-B,Hambly et al}. We model the factors as a multi-dimensional L\'evy diffusion and the underlying commodity prices are obtained by a linear mapping of the factors.  This makes our model more tractable yet keeping it still very flexible as it allows us to take, for example, heat rates and spreads into account in a natural way. Our study is also related to \cite{Keppo2004}, where a similar model is used to study hedging of swing options. We also refer to \cite{Kjaer}, where swing option pricing is considered under a non-Gaussian multi-factor price model. However, the analysis of \cite{Kjaer} is restricted to a modification of the so-called Deng model (see \cite{Deng}), which is a particular mean-reverting model. In our paper, we set up and analyze a class of models where the underlying factor prices follow a general L\'evy diffusion. The existing mathematical literature on swing options is mostly concerned with the pricing of a swing option. In addition to pricing, we also address the question of how to exercise a swing option optimally. From the analytical point of view, we identify using the HJB-equation an optimal exercise policy and characterize it in an intuitive way using the notion of marginal lost option value. We also present a numerical analysis of the problem including a numerical scheme based on the finite difference method.

%%% MORE ON NUMERICS

The reminder of the paper is organized as follows. In Section 2 we propose our model for the valuation of swing options. In Section 3 we analyze some general properties of the value function. Section 4 is devoted to the derivation of necessary and sufficient conditions for a function to coincide with the value function. We illustrate our results with explicit examples in Section 5, which are solved numerically in Section 6. Finally, we conclude in Section 7.

\section{The valuation model}
\subsection{The price dynamics} As we mentioned in the introduction, the prices of the commodities are driven by multiple factors. Throughout the study, the number of commodities is $m$ and the number of driving factors is $n$. The factor dynamics $X$ are modeled by an $n$-dimensional L\'evy diffusion. To make a precise statement, let $\left(\Omega,\mathcal{F},\mathbb{F},\mathbf{P}\right)$ be a complete filtered probability space satisfying the usual conditions, where $\mathbb{F}=\{\mathcal{F}\}_{t\geq0}$ is the filtration generated by $X$. We assume that the factor process $X$ are given as a strongly unique solution of the It\^o equation
\begin{equation}
\label{def:general process X}
dX(t) = \alpha(t,X(t))dt+\bar{\bar{\sigma}}(t,X(t))dW(t)+\int_{\mathbb{R}^l}\bar{\bar{\gamma}}(t,X(t),\xi)N(dt,d\xi),
\end{equation}
where $W=(W^{1},\ldots,W^{n_b})$ is an $n_b$-dimensional, potentially correlated, Brownian motion satisfying $d\langle W_{t}^{i},W_{t}^{j}\rangle=\rho_{ij}dt$ with $\rho_{ij}\in[-1,1]$ for all $i,j$. Furthermore, $N=(N^{1},\ldots,N^{n_l})$ denotes an $n_l$-dimensional Poisson random measure with L\'{e}vy measure $\nu$ given by the independent Poisson processes $\eta^{1},\ldots,\eta^{n_l}$.
Here, $\nu(\{0\})$ is the unit measure concentrated on zero and it is finite. The coefficients $\alpha: [0,T]\times\R^n\rightarrow\R^n$, $\bar{\bar{\sigma}}: [0,T]\times\R^n\rightarrow\R^{n\times n_b}$ and $\bar{\bar{\gamma}}: [0,T]\times\R^n\times\R^{n_l}\longrightarrow\R^{n\times n_l}$ are assumed to be sufficiently well behaving Lipschitz continuous functions to guarantee that the It\^o equation \eqref{def:general process X} has a unique strong solution -- see \cite{Applebaum}, p. 365 -- 366. The motivation to model the randomness using Brownian and finite activity jump noise
comes from electricity prices. In this framework, the jump process models the spiky behavior in the prices whereas the Brownian motion takes care of the small fluctuations.

Using the factor dynamics $X$, we define the $m$-dimensional price process $t\mapsto P_t:=P(X_t)$ via the linear transformation
\begin{equation}\label{Price Process}
P(x)=Bx,
\end{equation}
where $x\in\mathbb{R}^n$ and $B$ is a constant $m\times n$ matrix with $rank(B)=m\leq n$. In other words, there exists constants $(b_{ij})$ such
that $P^{i}(x)=\sum_{j=1}^{n}b_{ij}x^{j}$ for all $i=1,\dots,m$, that is, the commodity prices are linear combinations of the driving factors.
The component $P^i$ models the time evolution of the price of the $i$th commodity and this price is driven by the $n$
factors, i.e. the $n$-dimensional L\'{e}vy diffusion $X$ given as the solution of the It\^{o} equation \eqref{def:general process X}.
Since the price is linear as a function of the factors it is easy to change the model into a price model for spreads. Furthermore, the
matrix $B$ in \eqref{Price Process} can be interpreted as a constant weight between the different factors $X$ affecting the price. That allow us to take, for example, heat rates into account in our model.

In the definition of the factor dynamics, we assumed that the jump-diffusion $X$ and the driving Brownian motion and L\'{e}vy process have all different dimensions. For notational convenience, we assume in what follows that these dimensions are the same, i.e. $n=n_b=n_l$. We point out that the following analysis holds with obvious modifications also in the case where these dimensions are different.

\subsection{The valuation model}

The swing option written on the price process $P=P(X)$ gives the right to purchase the given amount $M$ of the commodities $i$ over the time period $[0,T]$. In addition to the global constraint $M$, the purchases are also subject to a local constraint $\bar{u}$ which corresponds to the maximal number of swing rights that can be exercised on a given time. Since the swing option can be exercised in continuous time, the local constraint is the maximum rate at which the option can be exercised. To formalize this, let $\mathcal{U}^i=\mathcal{U}_{M^i,\bar{u}^i}$ be the set of $\mathbb{F}$-measurable, real-valued processes $u^i=u^i(X)$ satisfying the constraints
\begin{displaymath}
u_t^i \in [0,\bar{u}^i], \ \int_0^T u^i_s ds \leq M^i,
\end{displaymath}
for all $i=1,\dots,m$ and $t\in[0,T]$. Here, the elements $\bar{u}=(\bar{u}^i)\in\mathbb{R}^m$ and $M=(M^i)\in\mathbb{R}^m$. The $\mathbb{R}^m$-valued process $Z$ defined as
\begin{equation}\label{def:Z}
Z^i_t = \int_0^t u_s ds,
\end{equation}
where $i=1,\dots,m$, keeps track of the amount purchased of commodity $i$ up to time $t$. In what follows, we call $Z$ the total volume and denote the product $\bigotimes_{i=1}^m\mathcal{U}^i$ as $\mathcal{U}$. The integral representation for $Z^{i}$ in \eqref{def:Z} is well defined due to the local constraint.

Denote the set $\mathcal{S}:=[0,T]\times\bigotimes_{i=1}^m[0,M_{i}] \times \R^m$ and define the affine function $A:\mathbb{R}^m\rightarrow\mathbb{R}^m$ as
\[ A(x)=Qx+K, \]
where $Q$ is an $m\times m$ matrix and $K\in\mathbb{R}^m$. Define the expected present value of the total exercise payoff  $J:\mathcal{S}\times\mathcal{U}\rightarrow\mathbb{R}$ given by the rate $u\in\mathcal{U}$ from time $t$ up to the terminal time $T$ (or, the performance functional of $u$) as
\begin{equation}
J(t,z,p,u)=\mathbf{E}\left[\int_t^T e^{-r(s-t)}\sum_{l=1}^m A^l(P_s)u^l_s ds \right. \left. \vphantom{\int_t^T} \right| \left. \vphantom{\int_t^T} Z_t=z, X_t=x \right],
\end{equation}
where $r>0$ is the constant discount factor. We point out that function $J$ is defined explicitly as a function of the factors $X$. This corresponds to that the holder of the contract observes the underlying factor and bases her exercise decisions of this information. Furthermore, we remark that this framework covers essentially call- and put-like payoffs, where the strike prices are given by the constant vector $K$. Now, the value function $V:\mathcal{S}\rightarrow\mathbb{R}$ is defined as
\begin{equation}\label{def:valuefunction}
V(t,z,p)=\sup_{u\in\mathcal{U}} J(t,z,p,u).
\end{equation}
We denote an optimal rate as $u^*$.

We make some remarks on the valuation problem \eqref{def:valuefunction}. The dimension of the decision variable $u$ is the same as the dimension of the price. That is, we can exercise the option for each price component, which corresponds to different commodities, with a different decision variable.
Furthermore, we defined the function $A$ such that it takes values in $\mathbb{R}^m$. This is done for notational convenience. Suppose that we have an $m$-dimensional price process but the decision variable $u$ is $k$-dimensional with $k\leq m$. This corresponds to the case where $m$ commodities are bundled into $k$ baskets and the holder can exercise the option on the baskets. Formally this is done by defining the affine function as $A:\R^{m}\rightarrow \R^{k}$. This will not affect the form of the value function.

\section{Some General Properties}

In this section we study some general properties of the valuation problem \eqref{def:valuefunction}. We split the analysis in two cases, depending on whether $M^i \geq \bar{u}^iT$ or $M^i < \bar{u}^iT$ for a given commodity $i$. In the latter case, the limit $M^i$ imposes an effective constraint on the usage of the option in the sense that the amount $M^i$ is dominated by the amount that can be purchased if the option is exercised on full rate over the entire time horizon. This case, i.e. the case when an effective volume constraint is present, is the interesting one from the practical point of view. It is also substantially more difficult to analyze mathematically as we will see later. Before considering this case, we study the complementary case when the effective volume constraint is absent. This will give us a point of reference in the other case.

\subsection{Without an effective volume constraint}
We consider first the case where $M^i\geq \bar{u}^i T$ for a given commodity $i$. The total volume constraint for the commodity $i$ is now superfluous, since it is possible for the holder to exercise the option at full rate throughout the lifetime of the contract. In the absence of an effective volume constraint for the commodity $i$, an optimal exercise rule is given by the next proposition.
\begin{proposition}\label{prop1}
Assume that $M^i\geq \bar{u}^i T$ for a given commodity $i$. Then an optimal exercise rate ${u^*}^i$ for the commodity $i$ reads as
\begin{equation*}
{u_{t}^*}^i=
\begin{cases}
\bar{u}^{i}, & \text{if } A^{i}(P_{t})>0, \\
0, & \text{if } A^{i}(P_{t})\leq 0,
\end{cases}
\end{equation*}
for all $t \in [0,T]$.
\end{proposition}

\begin{proof}
Let $u \in \mathcal{U}$ and $t\in [0,T]$. First, we observe that ${u^{*}}^i\in \mathcal{U}^i$. Furthermore, we find that
\begin{equation}\label{optimalnoconstraint}
\begin{split}
\mathbf{E}&\left[ \int_{t}^{T}e^{-r(s-t)}\sum_{l=1}^{m} A^{l}(P_{s})u_{s}^{l}ds \right. \left. \vphantom{\int_{t}^{T}} \right| \left. \vphantom{\int_{t}^{T}} Z_t=z, \ X_t=x\right]
\\&=\mathbf{E}\left[  \int_{t}^{T}e^{-r(s-t)}A^{i}(P_{t})u_{s}^{i}\mathbf{1}_{\{A^{i}(P_{t})\leq 0\}}ds\right. \left. \vphantom{\int_{t}^{T}} \right| \left. \vphantom{\int_{t}^{T}} Z_t=z, \ X_t=x\right]
\\&\quad+\mathbf{E}\left[ \int_{t}^{T}e^{-r(s-t)}A^{i}(P_{t})u_{s}^{i}\mathbf{1}_{\{A^{i}(P_{t})>0\}}ds\right. \left. \vphantom{\int_{t}^{T}} \right| \left. \vphantom{\int_{t}^{T}} Z_t=z, \ X_t=x\right]
\\&\quad+\mathbf{E}\left[\int_{t}^{T}e^{-r(s-t)}\sum_{\substack{l=1 \\ l\neq i}}^{m} A^{l}(P_{t})u_{s}^{l}ds\right. \left. \vphantom{\int_{t}^{T}} \right| \left. \vphantom{\int_{t}^{T}} Z_t=z, \ X_t=x\right]
\\&\leq \mathbf{E}\left[\int_{t}^{T}e^{-r(s-t)}A^{i}(P_{t}){u_{s}^{i}}^*ds\right. \left. \vphantom{\int_{t}^{T}} \right| \left. \vphantom{\int_{t}^{T}} Z_t=z, \ X_t=x\right]
\\&\quad+\mathbf{E}\left[\int_{t}^{T}e^{-r(s-t)}\sum_{\substack{l=1 \\ l\neq i}}^{m} A^{l}(P_{t})u_{s}^{l}ds\right. \left. \vphantom{\int_{t}^{T}} \right| \left. \vphantom{\int_{t}^{T}} Z_t=z, \ X_t=x\right].
\end{split}
\end{equation}
Now, take supremum over all $u^l$ on the left hand side and supremum over $u^l$, $l \neq i$, on the right hand side of \eqref{optimalnoconstraint}. Since the functional $J$ is linear in $u$, the same inequality still holds and, consequently, the conclusion follows.
 \end{proof}
Proposition \ref{prop1} states that in the absence of an effective volume constraint for commodity $i$, it is optimal to exercise the option whenever the payoff $A^i(P_t)$ is positive, i.e. when $(QP_t)^i\geq-K^i$. This is a natural result, since the holder does not have to worry of running out of the option over the planning horizon. Furthermore, since $\bar{u}$ is a constant vector, we find using Proposition \ref{prop1} that the value function does not depend on $z^i$ in the absence of an effective volume constraint for commodity $i$. This yields the following corollary.
\begin{corollary}[]
\label{lemma1}
In the absence of an effective volume constraint for a given commodity $i$, the marginal value $\frac{\partial V}{\partial z^{i}}(t,z,p)=0$.
\end{corollary}
Corollary \ref{lemma1} is also a very natural result. Indeed, if the holder uses the option on a commodity with no effective volume constraint, the option will not lose value.

To close the subsection, we discuss how the dimension of the range of the function $A$ affects the value given by \eqref{def:valuefunction}. For simplicity, assume that there is no effective volume constraint for any of the commodities and that the function $A:\mathbb{R}^m\rightarrow\mathbb{R}^m$ is of the form
\begin{equation}\label{DiagA}
A(x)=\diag(q_1,\dots,q_m)\cdot x + K,
\end{equation}
for $K\in\R^m$. Using Proposition \ref{prop1}, we know that the optimal exercise rule for the valuation problem specified by the payoff structure \eqref{DiagA} is
\begin{equation}
{u^*_{t}}^{l}=
\begin{cases}
\bar{u}^{l} & \text{if } A^{l}(P_{t})>0 \\
0 & \text{if } A^{l}(P_{t})\leq 0,
\end{cases}
\end{equation}
for all $t \in [0,T]$ and $l=1,\dots,m$. Formally, we can decrease the dimension of the range of $A$ from $m$, for example, as follows. Take $m'<m$ and define the $m'\times m$-matrix $\hat{Q}=(\hat{q}_{ij})$ such that each $q_i$ occurs only once and on exactly one column of $\hat{Q}$ and the other elements are zero. In financial terms, this means that the commodities $i$ are bundled into $m'$ pairwise disjoint baskets with weights $q_i$. Then the option gives exercise rights on each of these baskets with separate exercise rates. Now, let the function $\hat{A}:\mathbb{R}^m\rightarrow\mathbb{R}^{m'}$ be $\hat{A}(x)=\hat{Q}\cdot x +\hat{K}$ with $\hat{K}\in\R^{m'}$ such that
\begin{equation}\label{K:Cond}
\hat{K}_i = \sum_{\substack{j=1 \\ \hat{q}_{ij}\neq 0}}^{m} K_j,
\end{equation}
for all $i=1,\dots,m'$. Using the same reasoning as in Proposition \ref{prop1} we find that the optimal exercise rule for the valuation problem \eqref{def:valuefunction} given by $\hat{A}$ is
\begin{equation}
\hat{u}^l_{t}=
\begin{cases}
\hat{\bar{u}}^l & \text{if } \hat{A}^l(P_{t})>0 \\
0 & \text{if } \hat{A}^l(P_{t})\leq 0,
\end{cases}
\end{equation}
for all $t \in [0,T]$ and $l=1,\dots,m'$. Denote the value for $m$-dimensional ($m'$-dimensional) problem as $V$ ($\hat{V}$). Furthermore, denote the $m'$-dimensional total volume variable as $\hat{Z}$ and assume that all maximal exercise rates coincide: $\bar{u}^l=\hat{\bar{u}}^{l'}=\bar{u}$ for all $l=1,\dots,m$ and $l'= 1,\dots,m'$. Then, due to the structure of matrix $\hat{Q}$, we find using \eqref{K:Cond} that
\begin{equation}
\begin{split}
V(t,z,p)&=\mathbf{E}\left[\int_t^T e^{-r(s-t)} \sum_{l=1}^m A^l(P_s)\bar{u}\mathbf{1}_{\{q_lP^l_s+K_l>0 \}} ds \right. \left. \vphantom{\int_{t}^{T}} \right| \left. \vphantom{\int_{t}^{T}} Z_t=z, \ X_t=x\right]
        \\&\geq \mathbf{E}\left[\int_t^T e^{-r(s-t)} \sum_{i=1}^{m'}\left(\sum_{\substack{j=1 \\ \hat{q}_{ij}\neq 0}}^m (\hat{q}_{ij}P^j_s + \hat{K}_j)\bar{u}\mathbf{1}_{H^i_s} \right)ds \right. \left. \vphantom{\int_{t}^{T}} \right| \left. \vphantom{\int_{t}^{T}} \hat{Z}_t=\hat{z}, \ X_t=x\right] \\&= \hat{V}(t,\hat{z},p),
\end{split}
\end{equation}
where the events
\[ H^i_s =\left\{\sum_{\substack{j=1 \\ \hat{q}_{ij}\neq 0}}^m (\hat{q}_{ij}P^j_s + \hat{K}_j) >0 \right\}, \]
and the cumulative variable $\hat{Z}$ is defined analogously to \eqref{K:Cond}. Summarizing, we have shown that by bundling commodities $i$ into mutually disjoint baskets and, thus, reducing the dimension of the exercise rate process $u$, we lower the value of the option. This is, again, a natural result, since the bundling of commodities lowers flexibility of option contract in the sense that the holder must exercise the option at the same rate for all commodities in the same basket. This is in contrast to the case with separate commodities, where the exercise rates can be decided individually for each commodity.

\subsection{With an effective volume constraint}
In this section, we consider the case where $M_i\leq\bar{u}_i T$, in other words, the case when the total volume constraint is less than the maximal amount of commodity that can be acquired over the lifetime of the option. From the practical point of view, this is the interesting case. It is also substantially more difficult to analyze, since in this case we cannot find an optimal exercise policy explicitly as in Proposition $\ref{prop1}$. Instead we find the value function as the solution to the HJB-equation and an optimal exercise policy is obtained as a biproduct.

Our first task is to write the conditional expectation in \eqref{def:valuefunction} such that it depends explicitly on $Z$. This will be helpful in the later analysis. To this end, define the process $Y$ as $Y_t=e^{-rt}\sum_{l=1}^{m}A^{l}(P(X_{t}))Z_{t}^{l}$. Then the It\^{o} formula yields
\begin{align} \label{o0}
dY_t&=-re^{-rt}\sum_{l=1}^{m}A^{l}(P(X_{t}))Z_{t}^{l}dt
\\&+\sum_{i=1}^{n}e^{-rt}\sum_{l=1}^{m}\sum_{j=1}^{m}Z_{t}^{l}\frac{\partial A^{l}(P(X_{t}))}{\partial p^{j}}\frac{\partial P^{j}(X_{t})}{\partial
x^{i}}(\alpha_{i}(t,X_{t})dt+\sigma_{i}(t,X_{t})dW_{t})\nonumber
\\&+e^{-rt}\sum_{l=1}^{m}A^{l}(P(X_{t}))u_{t}^{l}dt \nonumber
\\&+\frac{1}{2}\sum_{i,k=1}^{n}(\sigma\sigma^{T})_{ik}e^{-rt}\sum_{l=1}^{m}Z_{t}^{l}\frac{\partial^{2}A^{l}(P(X_{t}))}{\partial x^{i}\partial x^{k}}d\langle W^{i},W^{k}\rangle \nonumber
\\&+\sum_{k=1}^{n}\int_{\R}e^{-rt}\sum_{l=1}^{m}Z_{t}^{l}\left\{A^{l}(P(X_{t^{-}}
+\gamma^{(k)}(t,X_{t^{-}},\xi^{k})))-A^{l}(P(X_{t^{-}}))\right\}N^{k}(dt,d\xi^{k})\nonumber.
\end{align}
Note that since $A$ is affine and $P$ is linear, we have
\begin{align}
\label{o1}
&A^{l}(P(X_{t^{-}}+\gamma^{(k)}(t,X_{t^{-}},\xi^{k})))-A^{l}(P(X_{t^{-}}))\nonumber\\
&=\sum_{v=1}^{m}q_{lv}P^{v}(X_{t^{-}}+\gamma^{(k)}(t,X_{t^{-}},\xi^{k}))+K^{l}-q_{lv}P^{v}(X_{t^{-}})-K^{l}
\nonumber \\
&=\sum_{v=1}^{m}q_{lv}P^{v}(\gamma^{(k)}(t,X_{t^{-}},\xi^{k})),
\end{align}
\begin{equation}
\label{o2}
\frac{\partial^{2}A^{l}(P(X_{t}))}{\partial x^{i}\partial x^{k}}=0,
\end{equation}
\begin{equation}
\label{o3}
\frac{\partial A^{l}(P(X_{t}))}{\partial p^{j}}=q_{lj},\quad \frac{\partial P^{j}(X_{t})}{\partial x^{i}}=b_{ji} \quad \text{and} \quad \sum_{j=1}^{m}q_{lj}b_{ji}=c_{li},
\end{equation}
where $q_{lj}$, $b_{ji}$ and $c_{li}$ are constants for all $i,j,l$. Substitution of \eqref{o1}, \eqref{o2} and \eqref{o3} into \eqref{o0} yields
\begin{align}
\label{firstito}
dY_t&=-re^{-rt}\sum_{l=1}^{m}A^{l}(P(X_{t}))Z_{t}^{l}dt+\sum_{i=1}^{n}e^{-rt}\sum_{l=1}^{m}Z_{t}^{l}c_{li}(\alpha_{i}(t,X_{t})dt+\sigma_{i}(t,X_{t})dW_{t}) \nonumber \\
&+e^{-rt}\sum_{l=1}^{m}A^{l}(P(X_{t}))u_{t}^{l}dt \nonumber \\
&+\sum_{k=1}^{n}\int_{\R}e^{-rt}\sum_{l=1}^{m}Z_{t}^{l}\sum_{v=1}^{m}q_{lv}P^{v}(\gamma^{(k)}(t,X_{t^{-}},\xi^{k}))N^{k}(dt,d\xi^{k}).
\end{align}
Since
\begin{equation}
\label{Ncompensated}
N(dt,d\xi)=\tilde{N}(dt,d\xi)+\nu(d\xi)dt,
\end{equation}
where $\nu$ is the L\'evy measure, we find that \eqref{firstito} can be written as
\begin{align*}
dY_t&=-re^{-rt}\sum_{l=1}^{m}A^{l}(P(X_{t}))Z_{t}^{l}dt+\sum_{i=1}^{n}e^{-rt}\sum_{l=1}^{m}Z_{t}^{l}c_{li}(\alpha_{i}(t,X_{t})dt+\sigma_{i}(t,X_{t})dW_{t}) \nonumber \\
&+e^{-rt}\sum_{l=1}^{m}A^{l}(P(X_{t}))u_{t}^{l}dt \nonumber \\
&+\sum_{k=1}^{n}\int_{\R}e^{-rt}\sum_{l=1}^{m}Z_{t}^{l}\sum_{v=1}^{m}q_{lv}P^{v}(\gamma^{(k)}(t,X_{t^{-}},\xi^{k}))\nu^{k}(d\xi^{k})dt\nonumber\\
&+\sum_{k=1}^{n}\int_{\R}e^{-rt}\sum_{l=1}^{m}Z_{t}^{l}\sum_{v=1}^{m}q_{lv}P^{v}(\gamma^{(k)}(t,X_{t^{-}},\xi^{k}))\tilde{N}^{k}(dt,d\xi^{k}).
\end{align*}
By integrating this from $t$ to $T$, we obtain
 \begin{align}
 \label{o4}
&e^{-rT}\sum_{l=1}^{m}A^{l}(P(X_{T}))Z_{T}^{l}-e^{-rt}\sum_{l=1}^{m}A^{l}(P(X_{t}))Z_{t}^{l}\nonumber\\
&=\int_{t}^{T}\Bigg[-re^{-rs}\sum_{l=1}^{m}A^{l}(P(X_{s}))Z_{s}^{l}+e^{-rs}\sum_{i=1}^{n}\sum_{l=1}^{m}Z_{s}^{l}c_{li}\alpha_{i}(s,X_{s})\nonumber\\
&+\sum_{k=1}^{n}\int_{\R}e^{-rs}\sum_{l=1}^{m}Z_{s}^{l}\sum_{v=1}^{m}q_{lv}P^{v}(\gamma^{(k)}(s,X_{s^{-}},\xi^{k}))\nu^{k}(d\xi^{k})\Bigg]ds\nonumber\\
&+\int_{t}^{T}e^{-rs}\sum_{l=1}^{m}A^{l}(P(X_{s}))u_{s}^{l}ds
+\int_{t}^{T}e^{-rs}\sum_{i=1}^{n}\sum_{l=1}^{m}Z_{s}^{l}c_{li}\sigma_{i}(s,X_{s})dW_{s}\nonumber \\
&+\int_{t}^{T}\sum_{k=1}^{n}\int_{\R}e^{-rs}\sum_{l=1}^{m}Z_{s}^{l}\sum_{v=1}^{m}q_{lv}P^{v}(\gamma^{(k)}(s,X_{s^{-}},\xi^{k}))\tilde{N}^{k}(ds,d\xi^{k}).
\end{align}
Consider first the Brownian integral $\int_{t}^{T}e^{-rs}\sum_{i=1}^{n}\sum_{l=1}^{m}Z_{s}^{l}c_{li}\sigma_{i}(s,X_{s})dW_{s}$. Each of the integrands is of the form $e^{-rs}Z_{s}^{l}c_{li}\sigma_{ij}(s,X_{s})dW_{s}^{j}$. By definition of $Z_{t}^{l}$, we know that $0\leq Z_{t}^{l}\leq \bar{u}^{l}t$. Since $Z_{t}^{l}$ is nondecreasing, it follows that $(Z_{t}^{l})^{2}\leq(\bar{u}^{l}t)^{2}\leq(\bar{u}^{l}T)^{2}$. Hence,
\begin{align*}
&\mathbf{E}\left[\int_{0}^{T}e^{-2rs}(Z_{s}^{l})^{2}c_{li}^{2}\sigma^2_{ij}(s,X_{s})ds\right. \left. \vphantom{\int_{0}^{t}} \right| \left. \vphantom{\int_{t}^{T}} Z_0=z, \ X_0=x\right]\nonumber\\
&\leq (\bar{u}^{l}Tc_{li})^{2}\mathbf{E}\left[\int_{0}^{T}\sigma^2_{ij}(s,X_{s})ds\right. \left. \vphantom{\int_{0}^{t}} \right| \left. \vphantom{\int_{t}^{T}} Z_0=z, \ X_0=x\right]<\infty.
\end{align*}
Using a martingale representation theorem, see, e.g. \cite{Applebaum}, Thrm. 5.3.6, we conclude that
\[ t\mapsto \int_{0}^{t}e^{-rs}\sum_{i=1}^{n}\sum_{l=1}^{m}Z_{s}^{l}c_{li}\sigma_{i}(s,X_{s})dW_{s} \]
is a martingale with respect to $\mathbb{F}$. Using the same argument, we find that the process %Furthermore, assume that for all $l,i,k$ we have
%\begin{equation}
%\label{compensatedmartingale}
%\mathbf{E}\left[\int_{0}^{T}\int_{\R}\left[e^{-rs}Z_{s}^{l}P^{i}(\gamma^{(k)}(s,X_{s},\xi^{k}))\right]^{2}\nu^{k}(d\xi^{k})dt\right]\leq \infty.
%\end{equation}
%Then the process
\begin{equation*}
%\label{compensatedmartingale}
t\mapsto\int_{0}^{t}\sum_{k=1}^{n}\int_{\R}e^{-rs}\sum_{l=1}^{m}Z_{s}^{l}\sum_{v=1}^{m}q_{lv}P^{v}(\gamma^{(k)}(s,X_{s^{-}},\xi^{k}))\tilde{N}^{k}(ds,d\xi^{k})
\end{equation*}
is a also a martingale with respect to $\mathbb{F}$. Consequently, the conditional expectation with respect to $\mathcal{F}_t$ is zero for the last two terms in \eqref{o4}.

By multiplying \eqref{o4} with $e^{rt}$ on both sides, substituting into $(\ref{def:valuefunction})$ and using the martingale properties, we find
\begin{align}
V(t,z,p)&=\sup_{u\in \mathcal{U}}\mathbf{E}\left[e^{-r(T-t)}\sum_{l=1}^{m}A^{l}(P(X_{T}))Z_{T}^{l}-\sum_{l=1}^{m}A^{l}(P(X_{t}))Z_{t}^{l} \right. \nonumber\\
&\left.-\int_{t}^{T}\left(-re^{-r(s-t)}\sum_{l=1}^{m}A^{l}(P(X_{s}))Z_{s}^{l}+e^{-r(s-t)}\sum_{i=1}^{n}\sum_{l=1}^{m}Z_{s}^{l}c_{li}\alpha_{i}(s,X_{s}) \right.\right.\nonumber\\
&+\left.\sum_{k=1}^{n}\int_{\R}e^{-r(s-t)}\sum_{l=1}^{m}Z_{s}^{l}\sum_{v=1}^{m}q_{lv}P^{v}(\gamma^{(k)}(s,X_{s-},\xi^{k}))\nu^{k}(d\xi^{k})\right)ds \left. \vphantom{\int_{t}^{T}} \right| \left. \vphantom{\int_{t}^{T}} Z_t=z, \ X_t=x\right].
\end{align}
Using that the measurability of $\sum_{l=1}^{m}A^{l}(P(X_{t}))Z_{t}^{l}$, we can express the value function \eqref{def:valuefunction} as
\begin{align*}
%\label{valufunctionexplicitZcomp}
V(t,z,p)&=-\sum_{l=1}^{m}A^{l}(p)z^{l}+\sup_{u\in \mathcal{U}}\mathbf{E}\left[e^{-r(T-t)}\sum_{l=1}^{m}A^{l}(P(X_{T}))Z_{T}^{l} \right. \nonumber\\
&-\int_{t}^{T}\left(-re^{-r(s-t)}\sum_{l=1}^{m}A^{l}(P(X_{s}))Z_{s}^{l}+e^{-r(s-t)}\sum_{i=1}^{n}\sum_{l=1}^{m}Z_{s}^{l}c_{li}\alpha_{i}(s,X_{s})\right.\nonumber\\
&+\left.\sum_{k=1}^{n}\int_{\R}e^{-r(s-t)}\sum_{l=1}^{m}Z_{s}^{l}\sum_{v=1}^{m}q_{lv}P^{v}(\gamma^{(k)}(s,X_{s^{-}},\xi^{k}))\nu^{k}(d\xi^{k})\right)ds\left. \vphantom{\int_{t}^{T}} \right| \left. \vphantom{\int_{t}^{T}} Z_t=z, \ X_t=x\right].
\end{align*}
We now have an explicit dependence on $Z$ in the value function, which will be useful in the proof of the following proposition. We point out that we can assume that we have an effective volume constraint in all commodities $i$, since the complementary case is already covered by Proposition \ref{prop1}.
\begin{proposition}[]
\label{vznonpositive}
In the presence of an effective volume constraint, i.e. when $M^{i}<\bar{u}^{i}T$, the marginal value $V_{z^{i}}(t,z,p)\leq 0$ for all $i$.
\end{proposition}

\begin{proof}
Let $u^{\varepsilon}=(u_{1}^{\varepsilon},\ldots,u_{m}^{\varepsilon}) \in \mathcal{U}^{\varepsilon}$ be processes giving rise to admissible exercise
policies $Z_{t}^{\varepsilon}=(Z_{t}^{1}+\varepsilon,\ldots,Z_{t}^{m}+\varepsilon)$ at time $t$. Let $u\in \mathcal{U}$ be the processes giving rise to
admissible exercise policies $Z_{t}$ at time $t$. Since the exercise policies $Z_{t}^{\varepsilon}$ arising from $u^{\varepsilon}$ are admissible
and must satisfy the effective volume constraint we have that $\mathcal{U}^{\varepsilon}\subseteq \mathcal{U}$. Also, for an arbitrary
admissible $s \mapsto Z_{s}^{\varepsilon}$ on $(t,T]$, define an associated $\check{Z}$ as
\begin{equation}
\label{checkZ}
\check{Z}_{s}=Z_{s}^{\varepsilon}-\varepsilon,
\end{equation}
for all $s \in (t,T]$. With this in mind, we proceed by expressing the marginal value as
\begin{eqnarray}
\frac{\partial V(t,z,p)}{\partial z^{j}}&=&-A^{j}(p) \nonumber \\
&&+\: \lim_{\varepsilon \to 0}\frac{1}{\varepsilon}\left\{\sup_{u^{\varepsilon}}\mathbf{E}\left[e^{-r(T-t)}\left(A^{j}(P(X_{T}))Z_{j}^{\varepsilon}(T)+\sum_{\substack{l=1 \\ l\neq j}}^{m}A^{l}(P(X_{T}))Z^{l}(T)\right) \right.\right. \nonumber \\
&&-\: \int_{t}^{T}\left[-re^{-r(s-t)}\left(A^{j}(P(X_{s}))Z_{j}^{\varepsilon}(s)+\sum_{\substack{l=1\\ l\neq j}}^{m}A^{l}(P(X_{s}))Z^{l}(s)\right) \right. \nonumber \\
&&+\: e^{-r(s-t)}\sum_{i=1}^{n}\left(Z_{j}^{\varepsilon}(s)c_{ji}+\sum_{\substack{l=1 \\ l\neq j}}^{m}Z^{l}(s)c_{li}\right)\alpha_{i}(s,X_{s})\nonumber\\
&&+\:\sum_{k=1}^{n}\int_{\R}e^{-r(s-t)}\left(Z_{j}^{\varepsilon}(s)+\sum_{\substack{l=1 \\ l\neq j}}^{m}Z^{l}(s)\right)\nonumber\\
&&\times\left.\sum_{v=1}^{m}q_{lv}P^{v}(\gamma^{(k)}(s,X_{s^{-}},\xi^{k}))\nu^{k}(d\xi^{k})\right]ds\left. \vphantom{\int_{t}^{T}} \right| \left. \vphantom{\int_{t}^{T}} Z_t=z, \ X_t=x\right] \nonumber \\
&&-\:\sup_{u}\mathbf{E}\left[e^{-r(T-t)}\sum_{l=1}^{m}A^{l}(P(X_{T}))Z^{l}(T) \right. \nonumber \\
&&-\: \int_{t}^{T}\left[-re^{-r(s-t)}\sum_{l=1}^{m}A^{l}(P(X_{s}))Z^{l}(s)+\: e^{-r(s-t)}\sum_{i=1}^{n}\sum_{l=1}^{m}Z^{l}(s)c_{li}\alpha_{i}(s,X_{s})\right.\nonumber\\
&&+\:\sum_{k=1}^{n}\int_{\R}e^{-r(s-t)}\sum_{l=1}^{m}Z^{l}(s)\nonumber\\
&&\times\left.\left.\sum_{v=1}^{m}q_{lv}P^{v}(\gamma^{(k)}(s,X_{s^{-}},\xi^{k}))\nu^{k}(d\xi^{k})\right]ds\left. \vphantom{\int_{t}^{T}} \right| \left. \vphantom{\int_{t}^{T}} Z_t=z, \ X_t=x\right]\right\}.
\end{eqnarray}
By collecting the terms containing $Z^{\varepsilon}$ in $\sup_{{u}^{\varepsilon}}$ and taking out the $j$th term in the supremum over $u$, we obtain
\begin{eqnarray}\label{vzcoll}
\frac{\partial V(t,z,p)}{\partial z^{j}}&=&-A^{j}(p) \nonumber\\
&&+\:\lim_{\varepsilon \to 0}\frac{1}{\varepsilon}\left\{\sup_{u^{\varepsilon}}\mathbf{E}\left[e^{-r(T-t)}A^{j}(P(X_{T}))Z_{j}^{\varepsilon}(T)\right.\right.\nonumber\\
&&-\:\int_{t}^{T}\left[-re^{-r(s-t)}A^{j}(P(X_{s}))Z_{j}^{\varepsilon}(s)+e^{-r(s-t)}\sum_{i=1}^{n}Z_{j}^{\varepsilon}(s)c_{ji}\alpha_{i}(s,X_{s})\right.\nonumber\\
&&+\:\left.\sum_{k=1}^{n}\int_{\R}e^{-r(s-t)}Z_{j}^{\varepsilon}(s)\sum_{v=1}^{m}q_{lv}P^{v}(\gamma^{(k)}(s,X_{s^{-}},\xi^{k}))\nu^{k}(d\xi^{k})\right]ds\left. \vphantom{\int_{t}^{T}} \right| \left. \vphantom{\int_{t}^{T}} Z_t=z, \ X_t=x\right]\nonumber\\
&&+\sup_{u^{\varepsilon}}I_{1}-\sup_{u}I_{1}\nonumber\\
&&-\sup_{u}\mathbf{E}\left[e^{-r(T-t)}A^{j}(P(X_{T}))Z^{j}(T)\right.\nonumber\\
&&-\:\int_{t}^{T}\left[-re^{-r(s-t)}A^{j}(P(X_{s}))Z^{j}(s)+e^{-r(s-t)}\sum_{i=1}^{n}Z^{j}(s)c_{ji}\alpha_{i}(s,X_{s})\right.\nonumber\\
&&+\:\left.\sum_{k=1}^{n}\int_{\R}e^{-r(s-t)}Z^{j}(s)\sum_{v=1}^{m}q_{lv}P^{v}(\gamma^{(k)}(s,X_{s^{-}},\xi^{k}))\nu^{k}(d\xi^{k})\right.]ds\left. \vphantom{\int_{t}^{T}} \right| \left. \vphantom{\int_{t}^{T}} Z_t=z, \ X_t=x\right]\Bigg\},
\end{eqnarray}
where
\begin{eqnarray}
I_{1}&:=&\mathbf{E}\left[e^{-r(T-t)}\sum_{\substack{l=1 \\ l\neq j}}^{m}A^{l}(P(X_{T}))Z^{l}(T) \right. \nonumber \\
&&-\: \int_{t}^{T}\left[-re^{-r(s-t)}\sum_{\substack{l=1\\ l\neq j}}^{m}A^{l}(P(X_{s}))Z^{l}(s)+e^{-r(s-t)}\sum_{i=1}^{n}\sum_{\substack{l=1 \\ l\neq j}}^{m}Z^{l}(s)c_{li}\alpha_{i}(s,X_{s})\right.\nonumber\\
&&+\sum_{k=1}^{n}\int_{\R}e^{-r(s-t)}\sum_{\substack{l=1 \\ l\neq j}}^{m}Z^{l}(s)\nonumber \\
&&\left. \times\sum_{v=1}^{m}q_{lv}P^{v}(\gamma^{(k)}(s,X_{s^{-}},\xi^{k}))\nu^{k}(d\xi^{k})\right]ds\left. \vphantom{\sum_{\substack{l=1 \\ l\neq j}}^{m}} \right| \left. \vphantom{\sum_{\substack{l=1 \\ l\neq j}}^{m}} Z_t=z, \ X_t=x\right].
\end{eqnarray}
Furthermore, define
\begin{eqnarray}
I_{\varepsilon}&:=&\mathbf{E}\left[e^{-r(T-t)}A^{j}(P(X_{T}))Z_{j}^{\varepsilon}(T)\right.\nonumber\\
&&-\:\int_{t}^{T}\left[-re^{-r(s-t)}A^{j}(P(X_{s}))Z_{j}^{\varepsilon}(s)+e^{-r(s-t)}\sum_{i=1}^{n}Z_{j}^{\varepsilon}(s)c_{ji}\alpha_{i}(s,X_{s})\right.\nonumber\\
&&+\:\sum_{k=1}^{n}\int_{\R}e^{-r(s-t)}Z_{j}^{\varepsilon}(s)\nonumber\\
&&\times\left.\sum_{v=1}^{m}q_{lv}P^{v}(\gamma^{(k)}(s,X_{s^{-}},\xi^{k}))\nu^{k}(d\xi^{k})\right]ds\left. \vphantom{\int_t^T} \right| \left. \vphantom{\int_t^T} Z_t=z, \ X_t=x\right],
\end{eqnarray}
and
\begin{eqnarray}
I_{0}&:=&\mathbf{E}\left[e^{-r(T-t)}A^{j}(P(X_{T}))(Z^{j}(T)+\varepsilon)\right.\nonumber\\
&&-\:\int_{t}^{T}\left[-re^{-r(s-t)}A^{j}(P(X_{s}))(Z^{j}(s)+\varepsilon)+e^{-r(s-t)}\sum_{i=1}^{n}(Z^{j}(s)+\varepsilon)c_{ji}\alpha_{i}(s,X_{s})\right.\nonumber\\
&&+\:\sum_{k=1}^{n}\int_{\R}e^{-r(s-t)}\left(Z^{j}(s)+\varepsilon\right)\nonumber\\
&&\times\left.\sum_{v=1}^{m}q_{lv}P^{v}(\gamma^{(k)}(s,X_{s^{-}},\xi^{k}))\nu^{k}(d\xi^{k})\right.]ds\left. \vphantom{\int_t^T} \right| \left. \vphantom{\int_t^T} Z_t=z, \ X_t=x\right].
\end{eqnarray}
Then we can write \eqref{vzcoll} as
\begin{eqnarray}
\frac{\partial V(t,z,p)}{\partial z^{j}}&=&-A^{j}(p)\nonumber \\
&&+\:\lim_{\varepsilon \to 0}\frac{1}{\varepsilon}\left\{\sup_{u^{\varepsilon}}I_{1}-\sup_{u}I_{1}+\sup_{u^{\varepsilon}}I_{\varepsilon}-\sup_{u}I_{0}+\varepsilon \mathbf{E}\left[e^{-r(T-t)}A^{j}(P(X_{T}))\right.\right.\nonumber\\
&&-\:\int_{t}^{T}\left[-re^{-r(s-t)}A^{j}(P(X_{s}))+e^{-r(s-t)}\sum_{i=1}^{n}c_{ji}\alpha_{i}(s,X_{s})\right.\nonumber\\
&&+\:\sum_{k=1}^{n}\int_{\R}e^{-r(s-t)}\sum_{v=1}^{m}q_{lv}\nonumber\\
&&\times\left.P^{v}(\gamma^{(k)}(s,X_{s^{-}},\xi^{k}))\nu^{k}(d\xi^{k})\right]ds\left. \vphantom{\int_t^T} \right| \left. \vphantom{\int_t^T} Z_t=z, \ X_t=x\right]\Bigg\}.
\end{eqnarray}
Since $U^{\varepsilon}\subseteq U$ we have that $\sup_{u^{\varepsilon}}I_{1}-\sup_{u}I_{1}\leq 0$. By \eqref{checkZ} there is an injective map between each functional $I_{\varepsilon}$ and $I_{0}$ for arbitrary $Z^{\varepsilon}$ such that $I_{\varepsilon} \hookrightarrow I_{0}$, hence $\sup_{u^{\varepsilon}}I_{\varepsilon}-\sup_{u}I_{0}\leq 0$. Consequently,
\begin{eqnarray}
\label{Vzineq}
\frac{\partial V(t,z,p)}{\partial z^{j}}&\leq&-A^{j}(p)+\mathbf{E}\bigg[e^{-r(T-t)}A^{j}(P(X_{T}))\nonumber\\
&&-\:\int_{t}^{T}\bigg[-re^{-r(s-t)}A^{j}(P(X_{s}))+e^{-r(s-t)}\sum_{i=1}^{n}c_{ji}\alpha_{i}(s,X_{s})\nonumber\\
&&+\:\sum_{k=1}^{n}\int_{\R}e^{-r(s-t)}\sum_{v=1}^{m}q_{lv}\nonumber\\
&&\times \left.P^{v}(\gamma^{(k)}(s,X_{s^{-}},\xi^{k}))\nu^{k}(d\xi^{k})\right]ds\left. \vphantom{\int_t^T} \right| \left. \vphantom{\int_t^T} Z_t=z, \ X_t=x\right].
\end{eqnarray}
By applying the It\^{o} formula to the process $s \mapsto e^{-rs}A^{j}(P(X_{s}))$ and taking conditional expectation with respect to $\mathcal{F}_{s}$, we find that the right-hand side of \eqref{Vzineq} is zero.
\end{proof}
In the case of a one-dimensional decision variable, i.e. an option of one commodity, this states an intuitively obvious result, namely that in the presence of an effective volume constraint, the usage of the option will lower its value. Note that if $M>\bar{u}T$ there is an $\varepsilon$ such that $\mathcal{U}^{\varepsilon}=\mathcal{U}$ and the map $(\ref{checkZ})$ is bijective. Hence, we obtain the result of Corollary $\ref{lemma1}$.

\section{The HJB-equation}

In the previous section, we studied the dynamic programming problem \eqref{def:valuefunction} first in the absence of an effective volume constraint for commodity $i$. We showed that in this case the optimal exercise rule can be determined explicitly and that the option does not lose value if used for this commodity. We also considered the problem in the presence of an effective volume constraint and showed that in this case it loses value when used. In this section, we determine an optimal exercise rule in the presence of an effective volume constraint. To this end, we first derive the associated HJB-equation.

For the reminder of the paper, we change the notation on the value function. Since $rank(B)=m$, from now on we may write $V$ explicitly as a function
of the factors $X$ instead of the price $P(x)=Bx$, that is, we write $V(t,z,x)$ instead of $V(t,z,p)$ where the domain of $V$ is modified accordingly.

\subsection{Necessary conditions} We derive now the HJB-equation of the problem \eqref{def:valuefunction}. To this end, assume that value $V$ exists. Then the Bellman principle of optimality yields
\begin{equation}\label{Bellman}
\begin{split}
V(t,z,x)=\sup_{u \in \mathcal{U}}\mathbf{E}&\left[\int_{t}^{w}e^{-r(s-t)}\sum_{l=1}^{m}A^{l}(P(X_{s}))u^{l}_{s}ds\right.\\&\left.+e^{-r(w-t)}V(w,Z_{w},X_{w}) \right.\left. \vphantom{\int_t^T} \right| \left. \vphantom{\int_t^T} Z_t=z, \ X_t=x\right],
\end{split}
\end{equation}
for all times $0 \leq t<w \leq T$. Rewrite the equation \eqref{Bellman} as
\begin{equation}\label{Bellman2}
\begin{split}
\sup_{u \in \mathcal{U}}\mathbf{E}&\left[\int_{t}^{w}e^{-rs}\sum_{l=1}^{m}A^{l}(P(X_{s}))u^{l}_{s}ds+e^{-rw}V(w,Z_{w},X_{w})\right.\\&\left.-e^{-rt}V(t,Z_{t},X_{t})\right.\left. \vphantom{\int_t^T} \right| \left. \vphantom{\int_t^T} Z_t=z, \ X_t=x\right]=0.
\end{split}
\end{equation}
Furthermore, assume that $V \in \mathcal{C}^{1,1,2}(\mathcal{S})$. Then we obtain by the It\^o formula
\begin{eqnarray}
\label{bellmanito}
\lefteqn{e^{-rw}V(w,Z_{w},X_{w})-e^{-rt}V(t,Z_{t},X_{t})=\int_{t}^{w}d(e^{-rs}V(s,Z_{s},X_{s})}\nonumber \\
&=&\int_{t}^{w}\big[ e^{-rs}V_{s}(s,Z_{s},X_{s})-re^{-rs}V(s,Z_{s},X_{s})+e^{-rs}\sum_{i=1}^{n}V_{x^{i}}(s,Z_{s},X_{s})\alpha_{i}(s,X_{s})\big]ds \nonumber \\
&&+\:\int_{t}^{w}e^{-rs}\frac{1}{2}\sum_{i,j}^{n}(\sigma \sigma^{T})_{ij}V_{x^{i}x^{j}}(s,Z_{s},X_{s})d\langle W^{i},W^{j}\rangle_s \nonumber \\
&&+\:\int_{t}^{w}e^{-rs}\sum_{l=1}^{m}V_{z^{l}}(s,Z_{s},X_{s})u_{s}^{l}ds\nonumber \\
&&+\:\int_{t}^{w}e^{-rs}\sum_{i=1}^{n}\sum_{l=1}^{m}V_{x^{i}}(s,Z_{s},X_{s})\sigma_{i}(s,X_s)dW_{s}\nonumber\\
&&+\:\int_{t}^{w}\sum_{k=1}^{n}\int_{\R}e^{-rs}\left[V(s,Z_{s},X_{s}+\gamma^{(k)}(s,X_{s},\xi^{k}))-V(s,Z_{s},X_{s})\right]N^{k}(ds,d\xi^{k}).
\end{eqnarray}
Here, $\sigma_{i}dW_{s}\equiv \sum_{j}\sigma_{ij}dW^{j}_{s}$ and
\begin{equation*}
%\label{xgamma}
x+\gamma^{k}=(x^{1}+\gamma^{k}_{1},\ldots,x^{n}+\gamma^{k}_{n}),
\end{equation*}
where $\gamma^{k}_{j}$ is the $jk$:th element in the matrix $\bar{\bar{\gamma}}$. By compensating the Poissonian stochastic integral in \eqref{bellmanito}, we find under suitable $L^2$-assumptions on $\sigma$ and $V_{x_i}$, see \cite{Applebaum}, Thrm. 5.3.6, that the Brownian and compensated Poissonian integrals in \eqref{bellmanito} are martingales. Thus the equation \eqref{Bellman2} yields
\begin{eqnarray}
\label{bellman3}
0&=&\sup_{u \in \mathcal{U}}\mathbf{E}\left[\int_{t}^{w}e^{-r(s-t)}\left[V_{s}(s,Z_{s},X_{s})+\sum_{i=1}^{n}V_{x^{i}}(s,Z_{s},X_{s})\alpha_{i}(s,X_{s})\right.\right.\nonumber\\
&&\:+\frac{1}{2}\sum_{i,j}^{n}(\sigma \sigma^{T})_{ij}V_{x^{i}x^{j}}(s,Z_{s},X_{s})\rho_{ij}-rV(s,Z_{s},X_{s})\nonumber\\
&&\:+\sum_{k=1}^{n}\int_{\R}\left[V(s,Z_{s},X_{s}+\gamma^{(k)}(s,X_{s},\xi^{k}))-V(s,Z_{s},X_{s})\right]\nu^{k}(d\xi^{k})\nonumber\\
&&\:+\left.\sum_{l=1}^{m}\left(A^{l}(P(X_{s}))+V_{Z^{l}}(s,Z_{s},X_{s})\right)u_{s}^{l}\right] ds\left. \vphantom{\int_t^T} \right| \left. \vphantom{\int_t^T} Z_t=z, \ X_t=x\right].
\end{eqnarray}
Define the integro-differential operator $\mathcal{L}$ on $\mathcal{C}^{1,1,2}(\mathcal{S})$ as
\begin{eqnarray}
\label{bellmanoperatorcompx}
\mathcal{L}F(t,z,x)&=&F_{t}(t,z,x)+\sum_{i=1}^{n}\alpha_{i}(t,x)F_{x^{i}}(t,z,x)+\frac{1}{2}\sum_{i,j}^{n}(\sigma\sigma^{T})_{ij}F_{x^{i}x^{j}}(t,z,x)\rho_{ij}\nonumber\\
&&\:+\sum_{k=1}^{n}\int_{\R}\left[F(s,z,x+\gamma^{(k)}(s,x,\xi^{k}))-F(s,z,x)\right]\nu^{k}(d\xi^{k}),
\end{eqnarray}
and rewrite \eqref{bellman3} as
\begin{equation*}%\label{bellman4}
\begin{split}
0=\sup_{u \in \mathcal{U}}\mathbf{E}&\left[\frac{1}{w-t}\int_{t}^{w}e^{-r(s-t)}\left[(\mathcal{L}-r)V(s,Z_s,X_s)\right.\right.\\&+\left.\sum_{l=1}^{m}\left(A^{l}(P(X_{s}))+V_{z^{l}}(s,Z_{s},X_{s})\right)u_{s}^{l} \right]ds \left. \vphantom{\int_t^T} \right| \left. \vphantom{\int_t^T} Z_t=z, \ X_t=x\right].
\end{split}
\end{equation*}
Under appropriate conditions on $V$, see, e.g. \cite{Fleming Soner}, we can pass to the limit $w\downarrow t$ and obtain the HJB-equation
\begin{equation}
\label{HJBx}
\left(\mathcal{L}-r\right)V(t,z,x)+\sup_{u}\left\{\sum_{l=1}^{m}(A^{l}(P(x))+V_{z^{l}}(t,z,x))u^{l}(t)\right\}=0,
\end{equation}
where the $u$ varies over the set of $\R^m$-valued functions defined on $[0,T]$ satisfying the conditions
\begin{equation*}
0\leq u^l(t) \leq \bar{u}^l , \ \int_0^t u^l(s)=z^l, \ \int_0^T u^l(s)ds \leq M^l,
\end{equation*}
for all $l=1,\dots,m$ and $t\in[0,T]$.

We observe from the equation \eqref{HJBx} that the sign of quantity $A^{l}(P(x))+V_{z^{l}}(t,z,x)$, $l=1,\dots,m$, determines whether the option should be exercised or not. From economic point of view, this quantity has a natural interpretation. Indeed, for a given commodity $l$, the function $A^{l}(P(\cdot))$ gives the instantaneous exercise payoff whereas the function $V_{z^l}$ measures the marginal lost option value. If the payoff dominates the lost option value for a given point $(t,z,x)$ and commodity $l$, the option should exercised at the full rate. That is, for each commodity $l$, the option should exercised according to the rule
\begin{equation*}%\label{admissible}
\hat{u}_{t}^{l}=
\begin{cases}
\bar{u}^{l} & \text{if } A^{l}(P(x))>-V_{z^{l}}(t,z,x), \\
0 & \text{if } A^{l}(P(x))\leq-V_{z^{l}}(t,z,x).
\end{cases}
\end{equation*}
We also point out that this rule is in line with the case when there is no effective volume constraint. In this case, the marginal lost option value is zero and, consequently, the option is used every time it yields a positive payoff. In particular, we find that the presence of an effective volume constraint postpones the optimal exercise of the option for a given commodity $l$.

\subsection{Sufficient conditions} In this subsection we consider sufficient conditions for a given function to coincide with the value function \eqref{def:valuefunction}. These conditions are given by the following verification theorem.
\begin{theorem}\label{sufficientcondition}
Assume that a function $F:\mathcal{S}\longrightarrow \R$ satisfies the following conditions:
\begin{itemize}
\item[(i)] $F(T,\cdot,\cdot)\equiv 0$, $F\in \mathcal{C}^{1,1,2}(\mathcal{S})$,
\item[(ii)]$(\mathcal{L}-r)F(t,z,x)+\sum_{l=1}^{m}(A^{l}(P(x))+F_{z^{l}}(t,z,x))u_{t}^{l}\leq 0$ for all $(t,z,x) \in \mathcal{S}$ and $u \in \mathcal{U}$, where $\mathcal{L}$ is defined in \eqref{bellmanoperatorcompx},
\item[(iii)]The processes
 \begin{itemize}
\item[a)] $\theta\mapsto \int_{0}^{\theta}e^{-rs}\sum_{i=1}^{n}F_{x^{i}}(s,Z_{s},X_{s})\sigma_{i}(s,X_{s})dW_s$,
\item[b)] $\theta\mapsto \int_{0}^{\theta}\sum_{k=1}^{n}\int_{\R}e^{-rs}\left[F(s,Z_{s},X_{s}+\gamma^{(k)}(s,\xi^{k}))-F(s,Z_{s},X_{s})\right]\tilde{N}^{k}(ds,d\xi^{k})$,
\end{itemize}
are martingales with respect to $\mathbb{F}$.
\end{itemize}
Then $F$ dominates the value $V$. In addition, if there exist an admissible $\mathring{u}$ such that
\begin{eqnarray}
\label{optimalcondition}
\lefteqn{(\mathcal{L}-r)F(t,z,x)+\sup_{u}\left[\sum_{l=1}^{m}(A^{l}(p(x))+F_{z^{l}}(t,z,x))u_{t}^{l}\right]}\nonumber\\
&=&(\mathcal{L}-r)F(t,z,x)+\sum_{l=1}^{m}(A^{l}(p(x))+F_{z^{l}}(t,z,x))\mathring{u}_{t}^{l}=0,
\end{eqnarray}
for all $(t,z,x) \in \mathcal{S}$, then $\mathring{u}=u^{*}$ and the function $F$ coincides with the value $V$.
\end{theorem}

\begin{proof}
Let $u \in \mathcal{U}$ and $t \in [0,T]$. By applying the It\^o formula to the process $ t \mapsto e^{-rt}F(t,Z_{t},X_{t})$, we find in the same way as in \eqref{bellmanito} that
\begin{eqnarray*}
%\label{bellmanitoF}
\lefteqn{e^{-rT}V(T,Z_{T},X_{T})-e^{-rt}F(t,Z_{t},X_{t})=\int_{t}^{T}d(e^{-rs}F(s,Z_{s},X_{s}))}\nonumber \\
&=&\int_{t}^{T}\big[ e^{-rs}F_{s}(s,Z_{s},X_{s})-re^{-rs}F(s,Z_{s},X_{s})+e^{-rs}\sum_{i=1}^{n}F_{x^{i}}(s,Z_{s},X_{s})\alpha_{i}(s,X_{s})\big]ds \nonumber \\
&&+\:\int_{t}^{T}e^{-rs}\frac{1}{2}\sum_{i,j}^{n}(\sigma \sigma^{T})_{ij}(s,X_s)F_{x^{i}x^{j}}(s,Z_{s},X_{s})d\langle W^{i},W^{j}\rangle_s \nonumber \\
&&+\:\int_{t}^{T}e^{-rs}\sum_{l=1}^{m}F_{z^{l}}(s,Z_{s},X_{s})u_{s}^{l}ds\nonumber \\
&&+\:\int_{t}^{T}e^{-rs}\sum_{i=1}^{n}\sum_{l=1}^{m}F_{x^{i}}(s,Z_{s},X_{s})\sigma_{i}(s,X_s)dW_{s}\nonumber\\
&&+\:\int_{t}^{T}\sum_{k=1}^{n}\int_{\R}e^{-rs}\left[F(s,Z_{s},X_{s}+\gamma^{(k)}(s,X_{s},\xi^{k}))-F(s,Z_{s},X_{s})\right]N^{k}(ds,d\xi^{k}).
\end{eqnarray*}
By using the assumption (i), definition of the operator $\mathcal{L}$ and the equation \eqref{Ncompensated}, we obtain the equality
\begin{eqnarray*}
%\label{bellmanitoF2}
\lefteqn{0=e^{-rt}F(t,Z_{t},X_{t})+\int_{t}^{T}e^{-rs}\left(\mathcal{L}-r\right)F(s,Z_{s},X_{s})ds+\int_{t}^{T}e^{-rs}\sum_{l=1}^{m}F_{z^{l}}(s,Z_{s},X_{s})u_{s}^{l}ds}\nonumber \\
&&+\:\int_{t}^{T}e^{-rs}\sum_{i=1}^{n}F_{x^{i}}(s,Z_{s},X_{s})\sigma_{i}dW\nonumber\\
&&+\:\int_{t}^{T}\sum_{k=1}^{n}\int_{\R}e^{-rs}\left[F(s,Z_{s},X_{s}+\gamma^{(k)}(s,X_{s},\xi^{k}))-F(s,Z_{s},X_{s})\right]\tilde{N}^{k}(ds,d\xi^{k}).
\end{eqnarray*}
Conditioning up to time $t$ and the assumption (iii) yields
\begin{eqnarray*}
%\label{bellmanitoF2conditon}
0&=&e^{-rt}F(t,Z_{t},X_{t})+\mathbf{E}\left[\int_{t}^{T}e^{-rs}\left(\mathcal{L}-r\right)F(s,Z_{s},X_{s})ds\right.\left. \vphantom{\int_t^T} \right| \left. \vphantom{\int_t^T} Z_t=z, \ X_t=x\right]\nonumber\\
&&+\:\mathbf{E}\left[\int_{t}^{T}e^{-rs}\sum_{l=1}^{m}F_{z^{l}}(s,Z_{s},X_{s})u_{s}^{l}ds\right.\left. \vphantom{\int_t^T} \right| \left. \vphantom{\int_t^T} Z_t=z, \ X_t=x\right].
\end{eqnarray*}
By assumption (ii), we get
\begin{equation}
\label{bellmanitoF2conditonii}
0\leq e^{-rt}F(t,Z_{t},X_{t})-\mathbf{E}\left[\int_{t}^{T}e^{-rs}\sum_{l=1}^{m}A^{l}(P(X_{s}))u_{s}^{l}ds\right.\left. \vphantom{\int_t^T} \right| \left. \vphantom{\int_t^T} Z_t=z, \ X_t=x\right].
\end{equation}
for all $u \in \mathcal{U}$. Thus, the first claim follows. Now, if there exist an admissible $\mathring{u}$ such that \eqref{optimalcondition} holds, then we would get equality in \eqref{bellmanitoF2conditonii}, i.e. for all $\omega$
\begin{equation*}
%\label{bellmanitoF2conditonoptimal}
0=e^{-rt}F(t,Z_{t},X_{t})-\mathbf{E}\left[\int_{t}^{T}e^{-rs}\sum_{l=1}^{m}A^{l}(P(X_{s}))\mathring{u}_{s}^{l}ds\right.\left. \vphantom{\int_t^T} \right| \left. \vphantom{\int_t^T} Z_t=z, \ X_t=x\right].
\end{equation*}
The conclusion follows.
 \end{proof}

\section{Examples} In this section we consider three examples. These examples illustrate two main issues. First, we compare two one-factor models (Example 1 and Example 2) with underlying Ornstein-Uhlenbeck factor dynamics, where the first model has a single Brownian driver whereas the other is driven by a sum of a Brownian motion and a compound Poisson process. To illustrate the effect of the jumps, the parameters of the factor dynamics are fixed such that the volatilities and the long term means are matched. In the third example, we study a two-factor model with underlying Ornstein-Uhlenbeck factor dynamics. As we will observe, the boundary conditions in the factor price dimensions are a delicate matter in this case. In all these examples, the
%the affects of jumps and the increase in dimensionality. For simplicity we will consider one-factor models to illustrate the affects of jumps and a two-factor model to illustrate the increase in dimensionality. In all examples the price processes are OU-processes and the control is one-dimensional. In the first we only have Brownian motion with drift. In the second, we add a compound poisson process to the price process and match the total volatility with that without jumps. In the third example we consider a two-factor model. Due to the dimensionality the boundary conditions also becomes a delicate problem in example three. Our
aim is to find
\begin{equation}
\label{valufunctionex2and3}
V(t,z,x)=\sup_{u\in \mathcal{U}}\mathbf{E}\left[\int_{t}^{T}e^{-r(s-t)}(P(X_{s})-K)u_{s}ds\right.\left| \vphantom{\int_{t}^{T}} \right. \left. \vphantom{\int_{t}^{T}} Z_t=z, \ X_t=x \right].
\end{equation}
The boundary conditions in the $x$-direction are found by the same arguments as in \cite{BLN}. The terminal condition is
\begin{equation}
\label{terminalBC}
V(T,z,x)=0,
\end{equation}
for all $z$ and $x$. This follows directly from the definition of the value function. Furthermore, the boundary condition in the $z$-direction, i.e. $z=M$, is
\begin{equation}
\label{optionalityBC}
V(t,M,x)=0,
\end{equation}
for all $t$ and $x$. This follows from the fact that when $z=M$ the only exercise rule available in $\mathcal{U}$ is the trivial one. The conditions \eqref{terminalBC} and \eqref{optionalityBC} hold for all three examples below.

\subsection{Example 1}
Let the factor dynamics $X$ be given by
\begin{equation}\label{browniandrift}
dX_{s}=\kappa(\mu-X_{s})ds+\sigma dW_{s}, \quad X_{t}=x
\end{equation}
and $P(x)=x$, where $s>t$. Then it is well known that at time $s>t$, the solution
\begin{equation}
\label{OUexplicitex1}
X_{s}=(x-\mu)e^{-\kappa(s-t)}+\mu+\sigma\int_{t}^{s}e^{-\kappa(s-v)}dW_{v}, \quad X_{t}=x.
\end{equation}
Furthermore,
\begin{equation*}
%\label{ex1dist}
X_{s} \thicksim \mathcal{N}\left(\mu+e^{-\kappa(s-t)}(x-\mu), \frac{\sigma^{2}}{2\kappa}(1-e^{-2\kappa(s-t)})\right).
\end{equation*}
With this specification, the value function \eqref{valufunctionex2and3} is given as a solution to the HJB-equation
\begin{equation}
\label{hjbex1}
V_{t}(t,z,x)+\kappa(\mu-x)V_{x}(t,z,x)+\frac{1}{2}\sigma V_{xx}(t,z,x)-rV(t,z,x)+\sup_{u\in \mathcal{U}}[(x-K+V_{z}(t,z,x))u_{t}]=0,
\end{equation}
with boundary conditions conditions in $x$-direction given by
\begin{equation}
\label{BCxmaxex1}
V(t,z,x_{max})=\bar{u}\int_{t}^{\tau}e^{-r(s-t)}[(x_{max}-\mu)e^{-\kappa(s-t)}+\mu]ds,
\end{equation}
and
\begin{equation}
\label{BCxminex1}
V(t,z,x_{min})=\bar{u}\int_{\theta}^{T}e^{-r(s-t)}[(x_{min}-\mu)e^{-\kappa(s-t)}+\mu]ds.
\end{equation}
We remark that this is similar to the example in \cite{BLN}, Appendix A. However, we consider an arithmetic OU-process whereas in \cite{BLN} the dynamics are given by an \emph{exponential} OU-process.

\subsection{Example 2}
To illustrate the effect of the jumps in the factor dynamics, we add in this example a compound Poisson process to the factor dynamics defined in \eqref{browniandrift} and match the expectation and volatility with Example 1. More precisely, consider the factor dynamic given by the It\^{o} equation
\begin{equation*}
dX_{s}=\kappa(\tilde{\mu}-X_{s})ds+\tilde{\sigma} dW_{s}+dY_{s},
\end{equation*}
where the compound Poisson process $Y_{s}$ has L\'evy measure $\nu(dy)=f\alpha e^{-\alpha y}\mathbf{1}_{\{y\geq 0\}}dy$ with $f,\alpha>0$. This equation can be written as
\begin{equation}
\label{browniandriftjump}
dX_{s}=\kappa(\mu+\int_{0}^{\infty}y\nu(dy)-X_{s})ds+\tilde{\sigma} dW_{s}+\int_{\R}y\tilde{N}(dy,ds),
\end{equation}
where the compensator
\begin{equation*}
\int_{0}^{\infty}y\nu(dy)=\frac{f}{\alpha}.
\end{equation*}
With this specification, the value function \eqref{valufunctionex2and3} is given as a solution to the HJB-equation
\begin{eqnarray}
\label{hjbex2}
0&=&V_{t}(t,z,x)+\kappa(\tilde{\mu}-x)V_{x}(t,z,x)+\frac{1}{2}\tilde{\sigma} V_{xx}(t,z,x)+\int_{0}^{\infty}[V(t,z,x+y)-V(t,z,x)]f\alpha e^{-\alpha y}dy\nonumber\\
&&\:-rV(t,z,x)+\sup_{u\in \mathcal{U}}[(x-K+V_{z}(t,z,x))u_{t}],
\end{eqnarray}
%\emph{NOTE: The $\alpha$ in $(\ref{bellmanoperatorcompx})$ is still $(\mu-X_{s})$ and not $(\mu+\frac{f}{\alpha}-X_{s})$. Due to that we are looking at a compound poisson
%process, and not only the compensator of a compound poisson process.}\\
with boundary conditions in $x$-direction given by
\begin{equation}
\label{BCxmaxex2}
V(t,z,x_{max})=\bar{u}\int_{t}^{\tau}e^{-r(s-t)}[(x_{max}-\tilde{\mu})e^{-\kappa(s-t)}+\tilde{\mu}+\frac{f}{\kappa\alpha}(1-e^{-\kappa(s-t)})]ds,
\end{equation}
and
\begin{equation}
\label{BCxminex2}
V(t,z,x_{max})=\bar{u}\int_{\theta}^{T}e^{-r(s-t)}[(x_{max}-\tilde{\mu})e^{-\kappa(s-t)}+\tilde{\mu}+\frac{f}{\kappa\alpha}(1-e^{-\kappa(s-t)})]ds.
\end{equation}
%\emph{How to choose $\tilde{\mu}$ and $\tilde{\sigma}$.}\\
The solution $X_{s}$ to \eqref{browniandriftjump} is given by
\begin{equation}
\label{OUexplicitex2}
X_{s}=(x-\frac{f}{\alpha}-\tilde{\mu})e^{-\kappa(s-t)}+\tilde{\mu}+\frac{f}{\alpha}+\tilde{\sigma}\int_{t}^{s}e^{-\kappa(s-v)}dW_{v}+\int_{t}^{s}\int_{\R}e^{-\kappa(s-v)}y\tilde{N}(dy,dv).
\end{equation}
It is easy to compute from the expression above that
\begin{equation*}
%\label{meanex2}
\mathbf{E}[X_{s}]=(x-\frac{f}{\alpha}-\tilde{\mu})e^{-\kappa(s-t)}+\tilde{\mu}+\frac{f}{\alpha},
\end{equation*}
and
\begin{equation*}
\Var(X_{s})=\mathbf{E}[(X_{s}-\mathbf{E}[X_{s}])^2]=\frac{1}{2\kappa}(1-e^{-2\kappa(s-t)})(\tilde{\sigma}^{2}+\frac{2f}{\alpha^{2}}).
\end{equation*}
To match the volatility and the long term mean in Example 1 and Example 2, we solve the equations above for $\tilde{\mu}$ and $\tilde{\sigma}$, when the expectation and variance is equal to that in Example 1. It follows that
\begin{equation*}
%\label{mutilde}
 \tilde{\mu}=\mu-\frac{f}{\alpha},
\end{equation*}
and
\begin{equation*}
%\label{sigmatilde}
\tilde{\sigma}=\sqrt{\sigma^{2}-\frac{2f}{\alpha^{2}}}.
\end{equation*}
Then the mean and total volatility in Example 1 and Example 2 will be the same. This is good for comparison reasons, which will be discussed more in the next section.

\subsection{Example 3}
The purpose of this example is illustrate the results when the price is driven by multiple factors. To this end, let $\beta^{1}, \beta^{2},\lambda_{1}, \mu^{1}, f, \kappa, K$ be non-negative constants. Consider the two-factor model
\begin{equation*}
%\label{def: process X example 1}
dX(t) = \alpha(t,X(t))dt+\bar{\bar{\sigma}}(t,X(t))dW(t)+\int_{\mathbb{R}^2}\bar{\bar{\gamma}}(t,X(t),\xi)N(dt,d\xi),
\end{equation*}
where $dW_{t}=(dW^{1}_{t},0)$, $N(dt,d\xi)=(0,N^{2}(dt,d\xi^{(2)}))$ and the L\'evy measure is $\nu(dy)=(0,f\kappa e^{-\kappa y}\mathbf{1}_{\{y\geq0\}}dy)$. Here, $f$ is the jump frequency and $\kappa$ is the parameter of the exponentially distributed jumps. Furthermore,
\begin{equation*}
%\label{alphaex1}
\alpha(t,X_{t})=(\mu^{1}-\beta^{1}X^{1}_{t},-\beta^{2} X^{2}_{t}),
\end{equation*}
\begin{equation*}
%\label{sigmaex1}
\bar{\bar{\sigma}}(t,X(t))=\left(\begin{array}{cc}
              \lambda^{1} & 0 \\
	      0& 0\\
	      \end{array} \right),
\end{equation*}
\begin{equation*}
%\label{gammaex1}
\bar{\bar{\gamma}}(t,X(t))=\left(\begin{array}{cc}
              0 & 0 \\
	      0& \xi^{2}\\
	      \end{array} \right).
\end{equation*}
In component form we have,
\begin{equation*}
%\label{x1comp}
dX_{v}^{1}=(\mu^{1}-\beta^{1}X_{v}^{1})dv+\lambda^{1} dW_{v}^{1}, \quad X^{1}_{t}=x^{1},
\end{equation*}
and
\begin{equation*}
%\label{x2comp}
dX_{v}^{2}=-\beta^{2}X_{v}^{2}dv+\int_{R}\xi^{2}N^{2}(dv,d\xi^{2}), \quad X^{2}_{t}=x^{2}.
\end{equation*}
The solutions can be written as
\begin{equation}
\label{OUex1comp1explicit}
X_{s}^{1}=e^{-\beta^{1}(s-t)}x^{1}+\int_{t}^{s}\mu^{1}e^{-\beta^{1}(s-v)}+\int_{t}^{s}\lambda^{1}e^{-\beta^{1}(s-v)}dW_{v}^{1},
\end{equation}
and
\begin{equation}
\label{OUex1comp2explicit}
X_{s}^{2}=e^{-\beta^{2}(s-t)}x^{2}+\int_{t}^{s}\int_{\R}e^{-\beta^{2}(s-v)}\xi^{2}\tilde{N}^{2}(d\xi^{2},dv)+\int_{t}^{s}\int_{\R}e^{-\beta^{2}(s-v)}\xi^{2}\nu^{2}(dv,d\xi^{2}).
\end{equation}
To set up the valuation model, define the price function $P:\R^{2}\rightarrow \R$ as $P(x)=x^{1}+x^{2}$. Furthermore, let $Z$ be as in \eqref{def:Z} with $m=1$. The payoff is of call option type, i.e. $A(p)=p-K$. Then, the value function \eqref{def:valuefunction} reads as
\begin{equation*}
%\label{valuefunctionex1}
V(t,z,x)=\sup_{u\in \mathcal{U}}\mathbf{E}\left[ \int_{t}^{T}e^{-r(s-t)}(P(X_{t})-K)u_{s}ds\right.\left| \vphantom{\int_{t}^{T}} \right. \left. \vphantom{\int_{t}^{T}} Z_t=z, \ X_t=x \right].
\end{equation*}
From Proposition \ref{prop1}, we see that in the absence of an effective final volume constraint the optimal exercise policy $u^{*}$ is given by
\begin{equation*}
  %\label{optimalcontrolex1}
    u_{t}^{*}=
      \begin{cases}
      \bar{u} & \text{if } P(X_{t})>K, \\
       0 & \text{if } P(X_{t})\leq K.
      \end{cases}
    \end{equation*}
for all $t \in [0,T]$. Hence, it is optimal to use the option whenever the swing yields a positive payoff. This is in line with \cite{BLN}, in which no jumps are considered.

Consider now the case with an effective volume constraint. The value function can be written in the component form as
\begin{equation}
\label{valufunctionex1 in X}
V(t,z,x^{1},x^{2})=\sup_{u\in \mathcal{U}}\mathbf{E}\left[ \int_{t}^{T}e^{-r(s-t)}(X^{1}_{s}+X^{2}_{s}-K)u_{s}ds\right.\left| \vphantom{\int_{t}^{T}} \right. \left. \vphantom{\int_{t}^{T}} Z_t=z, \ X_t=x \right].
\end{equation}
This function is given as the solution to the HJB-equation
\begin{eqnarray}
\label{explicitHJBxex1}
\lefteqn{V_{t}(t,z,x^1,x^2)+(\mu^{1}-\beta^{1}x^{1})V_{x^1}(t,z,x^1,x^2)-\beta^{2} x^{2}V_{x^{2}}(t,z,x^1,x^2)}\nonumber\\
&&+\int_{\R}\left(V(t,z,x^1,x^2+\xi^{2})-V(t,z,x^1,x^2)\right)\nu^{2}(d\xi^{2}) \nonumber\\
&&+\:\frac{1}{2}\lambda^{2}V_{x^1x^1}(t,z,x^1,x^2)-rV(t,z,x^1,x^2)+\sup_{u\in \mathcal{U}}[(x^1+x^2-K+V_{z}(t,z,x^1,x^2))u_{t}]=0.
\end{eqnarray}
With boundary conditions in $t$- and $z$-direction
\begin{equation*}
%\label{terminalBCex1}
V(T,z,x^1,x^2)=0 \quad \text{and} \quad V(t,M,x^1,x^2)=0,
\end{equation*}
and boundary conditions in $x$-direction
\begin{equation}
\label{BCex1x1min}
V(t,z,x^1_{min},x^2),
\end{equation}
\begin{equation}
\label{BCex1x1max}
V(t,z,x^1_{max},x^2).
\end{equation}
To solve the problem \eqref{explicitHJBxex1}-\eqref{BCex1x1max}, we need to find the boundary conditions \eqref{BCex1x1min}-\eqref{BCex1x1max}. We assume that we only have positive finite jumps, i.e. $x^{2}\geq 0$ and that $0<x^1_{min}<<\mu^{1}$. That is, the problem is solved in the first quadrant in the $x^{1}x^{2}$-plane.
\begin{remark}[]
The reason for choosing these spatial boundaries is due to the properties of the HJB-equation.
%For the derivatives wrt $x^1$ it has the form of an inhogenuos heat equation, which require boundaries at infinity. For the derivatives wrt $x^2$ it has the form of an inhoegenus transport equation, which has solution of wave type propagating away from origo. (Im not really sure about this...but think we need to say something about it.)
In the $x^1$-direction we have diffusion, which requires boundary conditions at both ends. However, in the $x^2$-direction we have transport in the positive direction, because the coefficient in  front of $V_{x^2}$ is negative and that PDE is solved backward in time, and therefore no boundary condition is needed at $x^2=x^2_{max}$. At  $x^2=0$ the derivative in $x^2$-direction vanishes, thus no boundary condition is needed.
\end{remark}

In what follows, the calculations rely on the fact that the underlying factor dynamics are Ornstein-Uhlenbeck processes. By plugging in the processes $X^1$ and $X^2$ given by \eqref{OUex1comp1explicit} and \eqref{OUex1comp2explicit}, respectively, into the value function \eqref{valufunctionex1 in X} and rearranging the terms we obtain
\begin{eqnarray}
\label{valuefunction in X explicit}
V(t,z,x^1,x^2)&=&\sup_{u\in \mathcal{U}}\left\{\mathbf{E}\left[ \int_{t}^{T}e^{-r(s-t)}\left(x^{1}e^{-\beta^1(s-t)}+x^{2}e^{-\beta^2(s-t)}\right.\right.\right.\nonumber\\
&&+\:\left.\left.\left.\frac{f}{\kappa\beta^2}(1-e^{-\beta^2(s-t)})+\frac{\mu^1}{\beta^1}(1-e^{-\beta^1(s-t)})-K\right)u_{s}ds\right.\left| \vphantom{\int_{t}^{T}} \right. \left. \vphantom{\int_{t}^{T}} Z_t=z, \ X_t=x \right]\right.\nonumber\\
&&+\:\left.\mathbf{E}\left[ \int_{t}^{T}e^{-r(s-t)}\left(\int_{t}^{s}\lambda^1e^{-\beta^1(s-v)}dW_{v} \right.\right.\right. \\
&&+\left.\left.\left.\int_{t}^{s}\int_{\R}e^{-\beta^2(s-v)}\xi^2\tilde{N}^2(d\xi^2,dv)\right)u_{s}ds\right.\left| \vphantom{\int_{t}^{T}} \right. \left. \vphantom{\int_{t}^{T}} Z_t=z, \ X_t=x \right]\right\}\nonumber.
\end{eqnarray}
To compute $V(t,z,x^1_{max},x^2)$, we plug $x^1=x^1_{max}$ into \eqref{valuefunction in X explicit}. Since the volatilities of the processes $X^1$ and $X^2$ are not state dependent, we can, by choosing $x^1_{max}$ sufficiently large, expect the trajectories of the process $X^1+X^2$ to be decreasing until the maturity since both the processes $X^{1}$ and $X^{2}$ tend towards their long term means, $\mu^{1}$ and $0$, respectively. Then it is optimal to start to exercise the option immediately with maximum rate until $z=M$ since $x^1_{max}+x^2$ is much larger than the long time expectation. This argument holds since we only consider positive jumps, i.e. $x^1_{max}+x^2\geq x^{1}_{max}$. Thus, if we start the process $X^1+X^2$ in $(x^1_{max},x^2)$, we can define an optimal control as
\begin{equation}
\label{optimalcontrolmax}
u_{s}=\bar{u}\bold{1}_{s\in [t,t+\frac{M-z}{\bar{u}}]}(s).
\end{equation}
Then we get
\begin{eqnarray}
\label{valuefunction maxmax}
V(t,z,x^1_{max},x^2)&=&\bar{u}\int_{t}^{t+\frac{M-z}{\bar{u}}}e^{-r(s-t)}\left(x^{1}_{max}e^{-\beta^1(s-t)}+x^{2}e^{-\beta^2(s-t)}\right.\nonumber\\
&&+\:\left.\frac{f}{\kappa\beta^2}(1-e^{-\beta^2(s-t)})+\frac{\mu^1}{\beta^1}(1-e^{-\beta^1(s-t)})-K\right)ds.
\end{eqnarray}
Here, we used the fact that the control $u$ as defined in \eqref{optimalcontrolmax} is deterministic. This enables us to use the Fubini theorem and the martingale property for the second expectation in \eqref{valuefunction in X explicit}.

On the contrary, when we start the process $X^1$ at $x^1_{min}$, assumed to be sufficiently small, $X^1$ will increase until maturity. It is thus tempting to wait as long as possible before we use the control, cf. equation (4.6) in \cite{BLN}. However, since $x^{2}$ can be very large and the jump frequency is state independent, we are unable to draw this conclusion. To deal with this issue, we proceed as follows. First, we assume that the value function is continuous for all $(x^1,x^2)\in\R^2$ and both $x^1_{min}$ and $x^2_{max}$ are finite. Using this we consider the \emph{deterministic part} of the process $X^1+X^2$ starting at $(x^1_{min},x^2)$ at time $t$, that is, the first expectation in the value function \eqref{valuefunction in X explicit}. We observe that the integrand is a continuous function in time and will assume a maximum (and minimum) on $[t,T]$. Suppose it has its maximum at a time $t_{max}$. Furthermore, assume that
\begin{equation*}
%\label{Assumption}
\text{\emph{The deterministic part will dominate the whole process at $t_{max}$}.}
\end{equation*}
Then, due to continuity, the deterministic part will dominate the process on an interval $(t_1,t_2)$ that contains $t_{max}$. We then choose a control defined as
\begin{equation*}
%\label{assumptioncontrol}
u_{s}=\bar{u}\bold{1}_{s\in[t_1,t_2]}(s).
\end{equation*}
We substitute this into the expression \eqref{valuefunction in X explicit}. This is a deterministic control so, again, we can use the Fubini theorem and the martingale property to get rid of the conditional expectations. Define
\begin{align}
\label{deterministicpart}
J(t_1,t_2):=\bar{u}\int_{t_1}^{t_2}e^{-r(s-t)}&\left[x^{1}_{min}e^{-\beta^1(s-t)}+x^{2}e^{-\beta^2(s-t)}+\frac{f}{\kappa\beta^2}(1-e^{-\beta^2(s-t)})\right.\\
&\left.+\frac{\mu^1}{\beta^1}(1-e^{-\beta^1(s-t)})-K\right]ds\nonumber.
\end{align}
Then
\begin{equation}
\label{valuefunction minmin}
V(t,z,x^1_{min},x^2)=\max_{t_1,t_2\in [t,T]}J(t_1,t_2)
\end{equation}
which is found by solving the (deterministic) maximization problem:
\begin{align*}
%\label{maxJproblem}
\max_{t_1,t_2\in [t,T]}\bar{u}\int_{t_1}^{t_2}e^{-r(s-t)}&\left[x^{1}_{min}e^{-\beta^1(s-t)}+x^{2}e^{-\beta^2(s-t)}+\frac{f}{\kappa\beta^2}(1-e^{-\beta^2(s-t)})\right.\\&+\left.\frac{\mu^1}{\beta^1}(1-e^{-\beta^1(s-t)})-K
\right]ds\nonumber,
\end{align*}
subject to
\begin{equation}
\label{Jmaxconstraint}
\bar{u}(t_2-t_1)\leq M-z.
\end{equation}
We solve the limits $t_{1}, t_{2}$ numerically by using elementary calculus methods. That is, to solve the maximization problem, define
\begin{align*}
%\label{integrandinmaxproblem}
g(s,t):=e^{-r(s-t)}&\left[x^{1}_{min}e^{-\beta^1(s-t)}+x^{2}e^{-\beta^2(s-t)}+\frac{f}{\kappa\beta^2}(1-e^{-\beta^2(s-t)})\right.\\&\left.+\frac{\mu^1}{\beta^1}(1-e^{-\beta^1(s-t)})-K\right]\nonumber.
\end{align*}
This is the integrand in \eqref{deterministicpart}. By differentiating \eqref{deterministicpart} with respect to $t_1$ and $t_2$, we obtain the first order necessary conditions
\begin{equation}
\frac{\partial J(t_1,t_2)}{\partial t_2}=\bar{u}g(t_2,t)=0,
\end{equation}
and
\begin{equation}
\frac{\partial J(t_1,t_2)}{\partial t_1}=-\bar{u}g(t_1,t)=0.
\end{equation}
Furthermore, at the boundary where $t_2=t_1+\frac{M-z}{\bar{u}}$, we find
\begin{equation}
\frac{\partial J(t_1,t_2)}{\partial t_1}=\bar{u}(g(t_1+\frac{M-z}{\bar{u}},t)-g(t_1,t))=0.
\end{equation}

To conclude, we solve $t_1$ and $t_2$ from these three equations, substitute these into \eqref{deterministicpart} and see which gives the highest value still satisfying the constraint \eqref{Jmaxconstraint}.

\section{Numerical experiments}
In this section we will present the numerical solutions of the HJB--equations from the three examples in Section 5. All equations have the boundary conditions that the option value is zero when $t=T$ and $z=M$. In addition we truncate the boundary in infinity, and the truncated boundaries requires boundary conditions. In Example 1 we solve \eqref{hjbex1} with the boundary condition \eqref{BCxmaxex1} and \eqref{BCxminex1} on the truncated boundary. Similarly in Example 2 we solve \eqref{hjbex2} with the boundary condition \eqref{BCxmaxex2} and \eqref{BCxminex2} on the truncated boundary. And finally in Example 3 we solve \eqref{explicitHJBxex1} with the boundary condition \eqref{valuefunction maxmax} and \eqref{valuefunction minmin} on the truncated boundary. Below we specify the model parameters, which are chosen for the purpose of illustration, and the discretization parameters of the numerical scheme.

%\eqref{hjbex1} \eqref{BCxmaxex1} \eqref{BCxminex1}
%\eqref{hjbex2} \eqref{BCxmaxex2} \eqref{BCxminex2}
%\eqref{explicitHJBxex1} \eqref{BCex1x1min} \eqref{BCex1x1max}

\subsection{Numerical scheme}
In Example 1 and Example 2 the HJB equations are PDEs defined over the variables $t$, $z$ and $x$, while the HJB equation in Example 3 is defined over the variables  $t$, $z$, $x^1$ and $x^2$. The equations are solved with finite difference methods (FDM). We use a first order Euler scheme in  $t$-direction. The $z$- and $x^2$-directions are handled explicitly with first order upwind schemes, while the $x$ or $x^1$-directions are handled implicitly with a second order central difference scheme. We start the time stepping at $t=T$ and go backward in time until $t=0$. %At $t=T$ and for $z=M$ the price of the option is zero.

The domain is discretized with a uniform grid in the $t$-, $x^2$- and $z$-directions whereas in the $x$- or $x^1$-direction we use an adaptive grid. The integral term is approximated with numerical integration, more precisely, the rectangle method with second order midpoint approximations.
The truncated boundary in the direction of jump, i.e. the $x$-direction in Example 2 and the $x^2$- direction in Example 3, causes some problems for the approximation of the integral, which is supposed to have upper limits at infinity. This problem is solved by linearly extrapolating the option price outside the truncated domain, and integrating up to a level where we get sufficiently accurate approximation of the integral.

\subsection{Numerical examples}
In this subsection we study the numerical examples from the previous section. They are all motivated by some swing options traded in the Scandinavian electricity market, which are called "Brukstidskontrakt". Such a contract gives the owner the right to buy a certain amount of electricity for her own selection of hours during 1 year. More precisely, for each of the 8760 hours in 1 year the holder of the contract must choose whether or not to use the contract. In our example we set the portion to 50\%, i.e. the holder must choose 4380 hours.

In practice the contracts are usually paid in advance and not for each time it is exercised, so the strike price will be $K=0$ in all the examples. We also use $T=1$, $M=\frac 12$ and $\bar{u}=1$. The three examples are further specified in the following.

{\bf Example 1: }
The model is specified by $\kappa = 0.014, \mu=40, \sigma=2.36$.
The truncated domain is defined by $x_{min} = 18.7, x_{max} = 61.3$.
The discretization parameters are $\Delta t = \Delta z = \frac 1{1000}$. The grid in $x-$direction is adaptive and consists of 671 grid points, with higher grid point density around $\mu$ and lower density near the truncated boundaries, $x_{min}$ and $x_{max}$.

{\bf Example 2: }
The model is specified by $\kappa = 0.014,  \alpha = 0.4, f=0.04,  \mu=39.9, \sigma=2.3387$. With these parameters the model i Example 1 and Example 2 have thesame mean and volatility. The grid and the truncated domain is as in Example 2.

{\bf Example 3: }
The model is specified by $\beta^1 = 0.014,  \mu=40, \sigma = 2.36, \beta^2 = 0.04,  \kappa=0.014, f=0.04 $.
The truncated domain is defined by $x^1_{min} = 17.2, x^1_{max} = 62.8, x^2_{min} = 0, x^2_{max} = 9$.
The discretization parameters are $\Delta t = \frac 1{3200}, \Delta z = \frac 1{3198}, \Delta x^2 = \frac{x^2_{max}}{40}$. In the $x^1$-direction we use an adaptive of 1200 grid points, with higher grid point density around $\mu$ and lower density near the truncated boundary in $x^1$-direction, i.e. $x^1=x^1_{min}$ and $x^1=x^1_{max}$. On the truncated boundary condition in $x^1$-direction we need to calculate \eqref{valuefunction maxmax} and solve \eqref{valuefunction minmin}. With the parameters in our example, $t_1$ and $t_2$  in \eqref{valuefunction minmin} turn out to be $t_1=t$ and $t_2=t+\frac 12 -z$. The truncated boundary in $x^2$-direction needs no boundary condition due to the nature of the PDE.  We also tried adaptive grid in the $x^2$-direction, but this did not seem to improve the accuracy.

%\eqref{valuefunction maxmax}  \eqref{valuefunction minmin} is solved by choosing the interval limits as small as possible.

In the following we visualize the numerical solution of these three examples. The two things we are most interested in are the option prices and the trigger prices. The trigger prices are also referred to as exercise curves, and they tell us when to exercise and when to hold. The option price is a function of two variables for each point in time in Example 1 and Example 2, and can then be visualized in a 3D plot. However, the option price in Example 3 is a function of three variables and is therefore harder to visualize, even for a fixed point in time, but we present it in 3 plots. The trigger prices  are in \cite{BLN} presented as exercise curves (see more in figures below). These prices are presented similarly for the results of Example 1 and Example 2. But for Example 3 the trigger price is actually a 3D surface for each point in time. We solve this by projecting the surface down to two different planes. We could also have made 3D surface plots of the trigger price in this example, but we think 2D plots are more instructive.

Figure \ref{fig:1}  and Figure \ref{fig:2} show the option price  at $t=0.5$ for Example 1 and Example 2 respectively. We see that the two plots a quite similar, and that the option price increases with increasing $x$-values. This makes sense from an economical point of view, since one would expect that a higher spot price results in a higher option price. Mathematically we see it from the value function \eqref{valufunctionex2and3} and the fact that the solution functions \eqref{OUexplicitex1} and \eqref{OUexplicitex2} are increasing functions of $x$. Furthermore, the option price decreases with increasing $z$-values which is in accordance with Proposition \ref{vznonpositive}, stating that whenever using the option it loses value, which makes economical sense.

Figure \ref{fig:3} shows the difference in option price between Example 1 and Example 2 at the same time level. More precisely Figure \ref{fig:3} shows the price from Example 2 minus the price from Example 1. We see that this difference is negative, which means that the option price is a little higher when we assume an underlying jump process. This is reasonable since the Gaussian OU-process in Example 1 has a symmetric distribution whereas the non-Gaussian OU-process in Example 2 has a positively skewed distribution. This positive skewness, which increases value in financial markets, is caused by the fact that we only have positive jumps.
\begin{figure}[htbp]
\centering
\subfigure[Option price Example 1]{
	\includegraphics[width=0.4\textwidth]{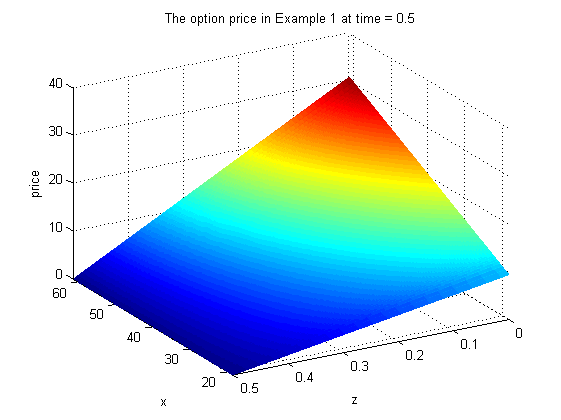}
	\label{fig:1}
}
\subfigure[Option price Example 2]{
	\includegraphics[width=0.4\textwidth]{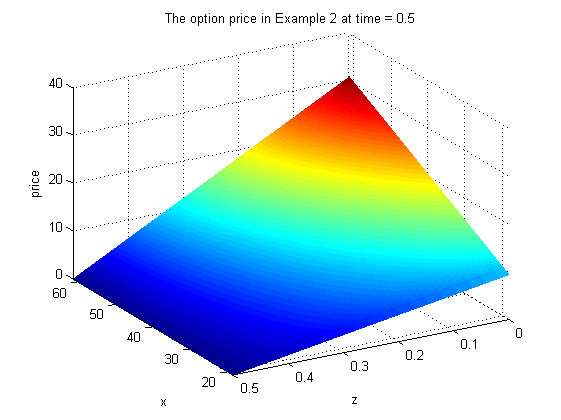}
	\label{fig:2}
}
\subfigure[Difference between option prices in Example 1 and Example 2]{
	\includegraphics[width=0.4\textwidth]{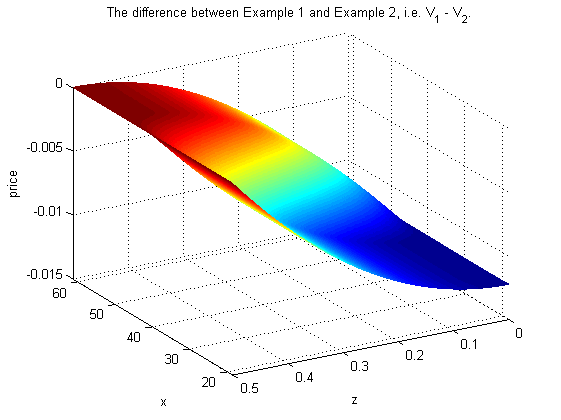}
	\label{fig:3}
}
\label{fig:123}
\caption{Option prices in Example 1 and Example 2}
\end{figure}

In Figure \ref{fig:4} we show the exercise curves for Example 1 and Example 2 at time = 0.5. The red curve corresponds to Example 2 and the black curve corresponds to Example 1. We see that the red curve lies more to the right than the black curve.

%, and this comes from the fact that the jumps gives higher expected prices, and therfore the holder will have higher prices before the option is exercised.
\begin{figure}[htbp]
\includegraphics[width=0.7\textwidth]{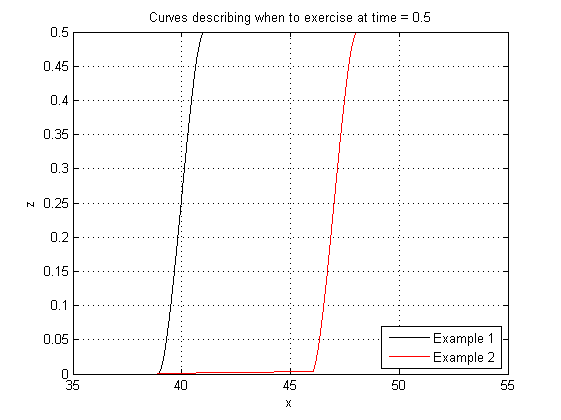}
\caption{Exercise curves for Example 1 and Example 2 at time = 0.5. Exercise when at points to the right of the curve and hold when at points to the left of the curve.}
\label{fig:4}
\end{figure}

%Before we continue the numerical investigation of the two-factor model in example 3, we would like to make a remark. The trigger price in example 1 and 2 is given by the levelcurve $\mathcal{e}=\{(t,z,x): V_{z}(t,z,x)+x=0\}$. The levelcurves $C=V_{z}(t,z,x)+x$, for any constant $C$, is shown in Figure \ref{fig:5}, from which we see that the gradiention in the x-direction of $V_{z}(t,z,x)+x$ is small. In particular, a significant shift of the levelcurve $\mathcal{e}$, will only cause a minor change in the option price. That is the option price  in Example 1 and Example 2 is not sensitive to the trigger price. Hence, even if the non-gaussian OU-process for the electricity spot price in example 2 causes a higher trigger price compared to the gaussian OU-process in Example 1 ( see Figure \ref{fig:4}), the resulting option price is almost the same, as one can see in Figure \ref{fig:3}. This indicate that in this particular examples the adding of jumps do not have any significant effect to the option price.
%\begin{figure}[htbp]
%\includegraphics[width=0.7\textwidth]{LevelCurves.png}
%\caption{Levelcurves for $x+V_z$ for example 1.}
%\label{fig:5}
%\end{figure}

%\begin{figure}[htbp]
%\includegraphics[width=0.7\textwidth]{figures\Exercise_ex1-ex2.png}
%\caption{Levelcurves for example 1.}
%\label{fig:5}
%\end{figure}
%\includegraphics[width=0.7\textwidth]{figures\LevelCurves.png}

In Figure  \ref{fig:8}--\ref{fig:10}  we show the option price of Example 3 at $t=0.5$. The function is plotted as a function of $x^1$ and $z$ for three values of $x^2$. We see that for each value of $x^2$ the plot looks similar to the plots in Figure \ref{fig:1} and Figure \ref{fig:2}, only the level of the surfaces changes a little.
\begin{figure}[htbp]
\centering
	\subfigure[$x^2=0$]{
	\includegraphics[width=0.45\textwidth]{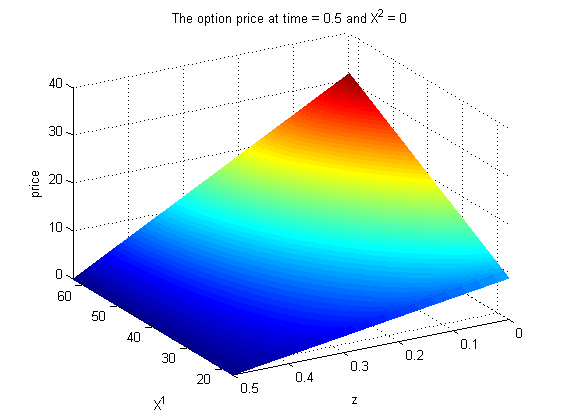}
	\label{fig:8}
	}
	\subfigure[$x^2=4.5$]{
	\includegraphics[width=0.45\textwidth]{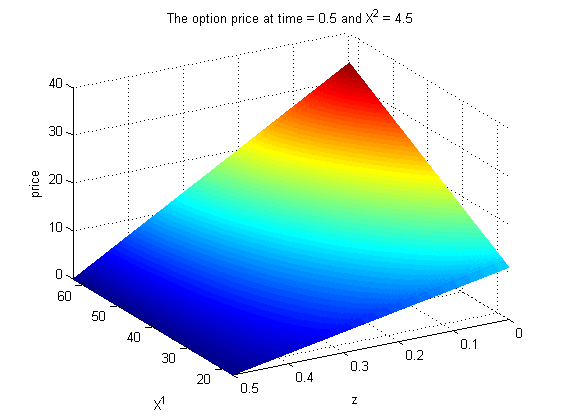}
	\label{fig:9}
	}
	\subfigure[$x^2=9$]{
	\includegraphics[width=0.5\textwidth]{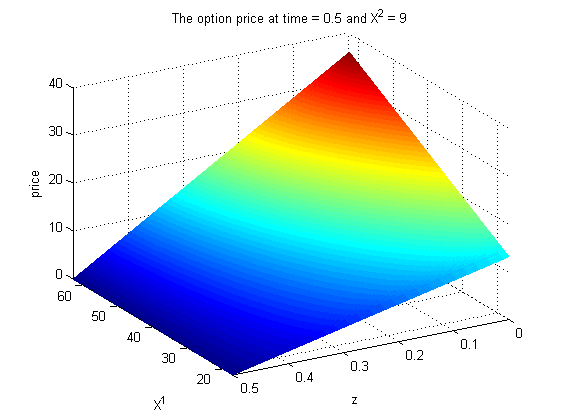}
	\label{fig:10}
	}
\label{fig:8910}
\caption{Option prices in Example 3. Plotted in the $x^1 z$-plane for $t=0.5$ and 3 values of $x^2$.}
\end{figure}

In Figure \ref{fig:5}--\ref{fig:6} we visualize the trigger price for $t=0.5$ in Example 3. Figure \ref{fig:5} shows the exercise surface projected down to the price,$z$-plane (where price = $x^1+x^2$), for various values of $x^2$. Figure \ref{fig:6} shows the exercise surface projected down to the  $x^1x^2$-plane, for various values of $z$.

It is worth noting that the slope of the curves in Figure \ref{fig:6} is approximately $-\frac{\beta_1}{\beta_2}=-0.35$. For example if we study the red curve and remove the point where $x^2=0$, and make a linear least squares approximation of it, it will have a slope of $-0.3517$. This means that the negative ratio of the mean reversion speeds approximates the slope of the exercise curves in the $x^1x^2$-plane. This is plausible from an economical point of view for the following reason. If the mean reversion speed is smaller for $x^1$ than for $x^2$, the holder will exploit a deviation from the long term mean earlier for process $X^2$ than $X^1$ by exercising the option. That is, she would require a higher contribution from $x^1$ to the price than from $x^2$ before exercising. This is because a high value of $x^2$ is likely to reduce more quickly and therefore it is beneficial to exercise with a lower value of $x^2$ in relation to $x^1$. On the contrary, if $x^1$ has a high price it is more likely to stay high longer. In this case, the holder might wait for even higher prices. A similar reasoning can be done for the opposite case when the mean reversion speed is bigger for $x^1$ than for $x^2$.

Notice that at $t=0.5$ and for any value of $z$, the holder should always hold if the values of $x^1$ and $x^2$ are small enough. The reason for this is the following: if the underlying price is small enough, you would expect it to be higher than this the rest of the time, and for $t=0.5$ it will therefore be beneficial to hold since $M=0.5$. However, in Figure \ref{fig:5} it can be seen that the exercise curves seems to stay above the line z=0.001. This is due to numerical error/instability. For finer grids the exercise curves will be closer to $z=0$ at for small values of $P=X^1+X^2$ and $t=0.5$. Similar effects can be observed for other values of $t$. The solution to this problem is either higher resolution on the grid, which requires more memory on the computer, or more accurate numerical schemes. The curves in Figure \ref{fig:6} bends a little in the lower right corner, and this is due to the mentioned instability in the numerical scheme.
\begin{figure}[htbp]
\centering
	\subfigure[Exercise curves plotted in the price,$z$-plane]{
	\includegraphics[width=0.4\textwidth]{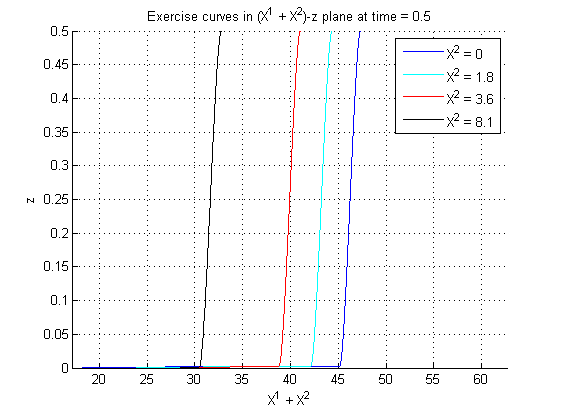}
	\label{fig:5}
	}
	\subfigure[Exercise curves plotted in the $x^1x^2$-plane]{
	\includegraphics[width=0.4\textwidth]{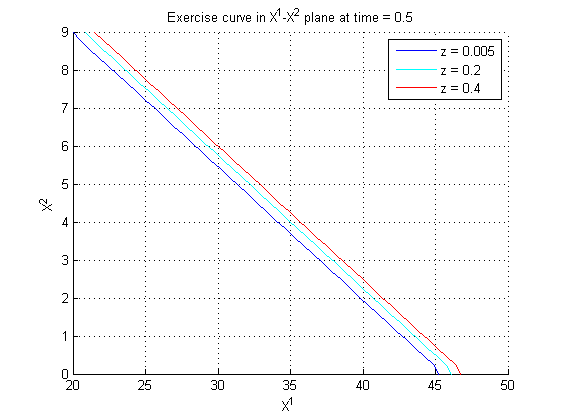}
	\label{fig:6}
	}
	\label{fig:56}
\caption{Trigger prices in Example 3}
\end{figure}

\subsection{Numerical accuracy in Example 3}
The analysis of the numerical scheme in Example 1 and Example 2 is similar to that of \cite{BLN}. To analyse Example 3 we study the Courant-Friedrichs-Lewy (CFL) condition.
The CFL number of the HJB equation in Example 3 is
\begin{align*}
C & = \frac {\Delta t}{\Delta x^2} \beta^2 x^2 + \frac {\Delta t}{\Delta z} \bar{u} \\
   & = \frac {\frac {1}{3200} }{\frac{9}{40}} 0.04 x^2 + \frac {\frac {1}{3200} }{\frac{1}{3198}}1\\
   & = \frac {1}{3200}   \left(1.6 \frac{x^2 }{9}+   3198  \right).
\end{align*}
A necessary condition for convergence is that the CFL number $C \leq C_{max}$, and in our case with $x^2 \in (0,9)$ we see that the CFL number $C \leq 1$. As mentioned we have observed some small instabilities in the numerical solution. We have also tried with larger values of $\Delta z$ compared to $\Delta t$, but the reported discretization parameters seem to give most accurate solutions. In order to establish convergence an implicit scheme should be developed, but this is not done in this work.

In the following we will present some evidence that the numerical solution converges to the correct solution of the HJB equation.
We attempt to evaluate both the numerical scheme and the calculated boundary conditions using \eqref{BCex1x1max}. This is done by trying to see how well they fit for extreme values of $x^1$, i.e. $x^1>>\mu$. We have used a very fine grid to solve HJB equation where $V$ is required to be linear in $x^1$-direction at the boundaries. This is similar to \cite{BLN}, where it is shown that this type of inaccurate boundary condition gives quite accurate solutions. Now this solution can be compared to the values we get from Equation \eqref{valuefunction maxmax}.

The difference between the numerical solution of the HJB equation and the value calculated by \eqref{valuefunction maxmax} is illustrated in Figure \ref{fig:7}. We see that the difference is relatively small compared to the option value.
%We see in Figure \ref{fig:7} that the difference between the numerical solution of the HJB equation and the value calculated by \eqref{valuefunction maxmax} is relatively small compared to the option value.
For this example with $t=0.5$ and $z=0.4$ it is about 6 orders of magnitude lower than the option value. This is sufficiently small for us to trust the numerical solver. The difference may come from all of the following five sources: assumptions that the control is as described in \eqref{optimalcontrolmax},  numerical inaccuracy/instability of the scheme, truncation in $x^1$ direction, truncation in $x^2$ direction, extrapolation in $x^2$ direction and linear boundary condition on the PDE solver.
%Notice also that the special behavior of the difference in Figure \ref{fig:7} for low values of $x^2$ which is due to the earlier mentioned numerical inaccuracy/instability.
\begin{figure}[htbp]
\includegraphics[width=0.7\textwidth]{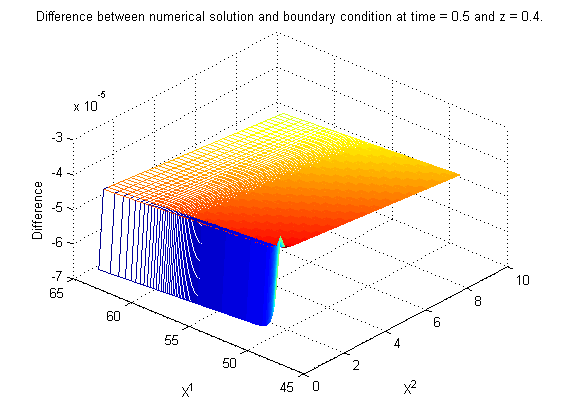}
\caption{Difference between the numerical solution and \eqref{BCex1x1max} for large values of $x^1$.}
\label{fig:7}
\end{figure}

\section{Conclusions}
In this paper, we developed and analyzed a valuation model for swing options on multi-commodity markets. The model is formulated as a dynamic programming problem, where the underlying dynamic structure is given by a multi-dimensional L\'{e}vy diffusion. This process models the price evolution of the commodities. The commodity prices are driven by a multi-dimensional Brownian motion and a multi-dimensional compound Poisson process. The introduction of the compound Poisson process is important since it allows non-Gaussian price evolution. This is important, in particular, on electricity markets, see, e.g. \cite{Benth Kallsen M-B, Hambly et al, Kjaer}. Furthermore, this model allows us to take into account jumps in price processes, which is also important on electricity markets.

From a analytical point of view, this study provides a multi-dimensional generalization of the analysis in \cite{BLN}. First, we analyze the model in the absence of an effective volume constraint. Along the lines of \cite{BLN}, we find that in this case the option does not loose value if used. Moreover, we prove that in the presence of an effective volume constraint for a given commodity, the usage of the option for this commodity will lower the value of the option. This is a intuitively appealing from the economical point of view. To tackle the problem of finding an optimal exercise rule and the price of the option, we analyze the pricing problem using the Bellman principle of optimality and derive the associated HJB-equation. In Section 4, we obtained an optimal exercise rule which states that if the immediate exercise payoff dominated the lost option value for a given commodity, then the option on this commodity should be exercised at a full rate. In particular, we conclude that this optimal exercise rule is a bang-bang rule. We also provide a verification theorem, which states conditions under which a given function coincides with the value function.

In addition we illustrate the results with three examples which we study numerically. We set up a straightforward FDM scheme to solve the associated HJB-equations. The numerical experiments seems to give reasonable results from both mathematical and economical points of view. In  the last of our examples we have also given evidence for convergence of the numerical solution. %In particular, we discussed in detail the delicate issue of boundary conditions in price dimensions for a two-factor model.

%- We have presented three numerical examples to illustrate the solution of the HJB and discussed boundary conditions.\\
%- The numerical results seems intuitive.\\
%-$\Delta t$ can be given an interpretation where it is related to the frequency of exercise opportunities of the original problem.\\

\subsubsection*{Acknowledgements}
Financial support from the project "Energy markets: modelling, optimization and simulation (EMMOS)", funded by the Norwegian Research Council under grant 205328 is gratefully acknowledged.

\end{document}